\documentclass[reqno, final,a4paper]{amsart}

\usepackage{color}
\usepackage[colorlinks=true,allcolors=blue
]{hyperref}
\usepackage{amsmath, amssymb, amsthm}
\usepackage{mathrsfs}
\usepackage{mathtools}
\usepackage[noabbrev,capitalize,nameinlink]{cleveref}
\crefname{equation}{}{}
\usepackage{fullpage}
\usepackage[noadjust]{cite}
\usepackage{graphics}
\usepackage{pifont}
\usepackage{tikz}
\usepackage{bbm}
\usepackage[T1]{fontenc}
\usetikzlibrary{arrows.meta}

\usepackage{environ}
\usepackage{framed}
\usepackage{url}
\usepackage[linesnumbered,ruled,vlined]{algorithm2e}
\usepackage[noend]{algpseudocode}
\usepackage[labelfont=bf]{caption}
\usepackage{framed}
\usepackage[framemethod=tikz]{mdframed}
\usepackage{pgfplots}
\usepackage{appendix}
\usepackage{graphicx}
\usepackage[textsize=tiny]{todonotes}
\usepackage{tcolorbox}
\usepackage{xcolor}
\usepackage[shortlabels]{enumitem}
\usepackage{cancel}
\usepackage[normalem]{ulem}
\crefformat{enumi}{#2#1#3}
\crefrangeformat{enumi}{#3#1#4 to~#5#2#6}
\crefmultiformat{enumi}{#2#1#3}
{ and~#2#1#3}{, #2#1#3}{ and~#2#1#3}
\allowdisplaybreaks

\let\originalleft\left
\let\originalright\right
\renewcommand{\left}{\mathopen{}\mathclose\bgroup\originalleft}
\renewcommand{\right}{\aftergroup\egroup\originalright}

\makeatletter
\newcommand*{\transpose}{%
  {\mathpalette\@transpose{}}%
}
\newcommand*{\@transpose}[2]{%
  \raisebox{\depth}{$\m@th#1\intercal$}%
}
\makeatother

\crefname{algocf}{Algorithm}{Algorithms}

\crefname{equation}{}{} 
\AtBeginEnvironment{appendices}{\crefalias{section}{appendix}} 

\colorlet{refkey}{orange!20}
\colorlet{labelkey}{blue!30}

\crefname{algocf}{Algorithm}{Algorithms}

\numberwithin{equation}{section}
\newtheorem{theorem}{Theorem}[section]

\newtheorem{lemma}[theorem]{Lemma}

\crefname{claim}{Claim}{Claims}

\newtheorem{corollary}[theorem]{Corollary}

\newtheorem*{question*}{Question}
\newtheorem{fact}[theorem]{Fact}

\theoremstyle{definition}
\newtheorem{definition}[theorem]{Definition}

\newtheorem*{definition*}{Definition}
\newtheorem*{theorem*}{Theorem}

\theoremstyle{remark}
\newtheorem{remark}[theorem]{Remark}
\newtheorem*{remark*}{Remark}

\newcommand{\mb}{\mathbb}

\newcommand{\mc}{\mathcal}

\newcommand{\mr}{\mathrm}

\newcommand{\on}{\operatorname}

\newcommand{\tr}{\operatorname{tr}}
\newcommand{\Zt}{\widetilde{Z}}

\newcommand{\eps}{\varepsilon}
\newcommand{\E}{\mathbb{E}}
\renewcommand{\Pr}{\mb P}
\renewcommand{\P}{\mathbb{P}}

\allowdisplaybreaks

\title{Entangled states are typically incomparable}

\author[Jain]{Vishesh Jain}

\address{University of Illinois Chicago}
\email{visheshj@uic.edu}

\author[Kwan]{Matthew Kwan}

\address{Institute of Science and Technology Austria (ISTA).}
\email{matthew.kwan@ist.ac.at}

\author[Michelen]{Marcus Michelen}

\address{Northwestern University}
\email{michelen@northwestern.edu}

\thanks{
Vishesh Jain was partially supported by NSF CAREER award DMS-2237646. Matthew Kwan was supported by ERC Starting Grant ``RANDSTRUCT'' No.\ 101076777. Marcus Michelen was supported in part by NSF grants DMS-2137623 and DMS-2246624.
}

\begin{document}

\maketitle
\begin{abstract}
Consider a bipartite quantum system, where Alice and Bob jointly possess a pure state $|\psi\rangle$. Using local quantum operations on their respective subsystems, and unlimited classical communication, Alice and Bob may be able to transform $|\psi\rangle$ into another state $|\phi\rangle$. Famously, \emph{Nielsen's theorem} [Phys.\ Rev.\ Lett., 1999] provides a necessary and sufficient algebraic criterion for such a transformation to be possible (namely, the \emph{entanglement spectrum} of $|\phi\rangle$ should \emph{majorise} the entanglement spectrum of $|\psi\rangle$).

In the paper where Nielsen proved this theorem, he conjectured that in the limit of large dimensionality, for \emph{almost all} pairs of states $|\psi\rangle, |\phi\rangle$ (according to the natural unitary invariant measure) such a transformation is not possible. That is to say, typical pairs of quantum states $|\psi\rangle, |\phi\rangle$ are entangled in fundamentally different ways, that cannot be converted to each other via local operations and classical communication.

Via Nielsen's theorem, this conjecture can be equivalently stated as a conjecture about majorisation of spectra of random matrices from the so-called \emph{trace-normalised complex Wishart--Laguerre ensemble}. Concretely, let $X$ and $Y$ be independent $n \times m$ random matrices whose entries are i.i.d.\ standard complex Gaussians; then Nielsen's conjecture says that the probability that the spectrum of $X X^\dagger / \tr(X X^\dagger)$ majorises the spectrum of $Y Y^\dagger / \tr(Y Y^\dagger)$ tends to zero as both $n$ and $m$ grow large. We prove this conjecture, and we also confirm some related predictions of Cunden, Facchi, Florio and Gramegna [J.~Phys.~A,~2020; Phys.~Rev.~A,~2021].
\end{abstract}

\section{Introduction}

Suppose a bipartite pure quantum state is distributed to two parties Alice and Bob. Imagine that Alice and Bob then travel to distant labs; after this point, they may only perform quantum operations on their own subsystem, though they are free to coordinate their operations via classical communication (e.g., they can share the results of any measurements). This paradigm is known as LOCC (local operations with classical communication), and plays a fundamental role in quantum information theory. See the Nielsen--Chuang monograph~\cite{NC00} for a general introduction to quantum information theory, and see the survey of Chitambar, Leung, Man\v cinska, Ozols, and Winter~\cite{CLMOW14} for a thorough introduction to the LOCC paradigm.

The LOCC paradigm naturally defines a partial order (called \emph{LOCC-convertibility} or \emph{entanglement transformation}) on the set of bipartite pure states $\mb C^{n}\otimes\mb C^{m}$ for any dimensions $n,m\in\mb N$. Namely, for a pair of states $|\psi\rangle,|\phi\rangle\in \mb C^{n}\otimes\mb C^{m}$, we write $|\psi\rangle\to |\phi\rangle$
(read ``$|\psi\rangle$ transforms to $|\phi\rangle$'')
if and only if, starting from the state $|\psi\rangle$, it is possible for Alice and Bob to coordinate some sequence of local operations to guarantee that they end up with the state $|\phi\rangle$.

In 1999, Nielsen~\cite{Nie99} found a beautiful connection between LOCC-convertibility and the mathematical theory of \emph{majorisation}. For the precise statement, and more context, the reader should refer to \cite{Nie99} (and Chapter~12.5.1 of Nielsen--Chuang~\cite{NC00}), but as a brief summary: although Alice and Bob share a pure state $|\psi\rangle$, Alice can interpret her subsystem as a \emph{mixed} state (a statistical ensemble of pure states in $\mb C^n$). This mixed state can be described by a \emph{density operator} $\rho_\psi=\on{tr}_2(|\psi\rangle\langle\psi|)$
(a Hermitian linear operator $\mb C^n\to \mb C^n$ obtained via a \emph{partial trace}, ``tracing out'' Bob's subsystem). The eigenvalues of $\rho_\psi$ are all nonnegative real numbers, summing to 1\footnote{We could equally well have defined these eigenvalues by tracing out Alice's subsystem, instead of Bob's subsystem. It is not hard to see that this choice only affects the number of zero eigenvalues; the nonzero eigenvalues do not actually depend on which of the two subsystems we choose to trace out.}. Nielsen's theorem is that LOCC-convertibility $|\psi\rangle\to |\phi\rangle$ is equivalent to the property that the spectrum of $\rho_\phi$ \emph{majorises} the spectrum of $\rho_\psi$: the sum of the $k$ smallest\footnote{Note that since the eigenvalues sum to $1$, this is equivalent to the condition that the sum of the $k$ \emph{largest} eigenvalues of $\rho_\psi$ is at most the sum of the $k$ \emph{largest} eigenvalues of $\rho_\phi$, for all $k\le n$.} eigenvalues of $\rho_\psi$ should be at most the sum of the $k$ smallest eigenvalues of $\rho_\phi$, for all $k \leq n$.  

Nielsen's theorem has a number of consequences. In particular, it lays bare the fact that (if $n,m\ge 3$) there are fundamentally different types of entanglement that cannot be LOCC-converted to each other: there exists a pair of states $|\psi\rangle,|\phi\rangle\in \mb C^{n}\otimes\mb C^{m}$ such that $|\psi\rangle\not\to|\phi\rangle$ and $|\phi\rangle\not\to|\psi\rangle$. Nielsen conjectured that this phenomenon is actually \emph{typical}: if $|\psi\rangle, |\phi\rangle$ are independent random states (both sampled from the Haar measure, with respect to unitary transformations, on the unit sphere in $\mb C^{n}\otimes\mb C^{m}$), then the probability of the event $|\psi\rangle\to |\phi\rangle$ tends to zero as $n,m\to \infty$. This conjecture has remained unproved for the last 25 years, though it is supported by strong numerical evidence (cf.\ the very convincing computations of Cunden, Facchi, Florio and Gramegna~\cite{CFFG20}). There are also quite convincing heuristic arguments for why it should be true: a brief heuristic argument was given by Nielsen~\cite{Nie99}, and another heuristic argument via integral geometry was given by \.{Z}yczkowski and Bengtsson~\cite{ZB02} (see also their monograph \cite{ZBbook} on the geometry of entanglement). Also, an \emph{infinite-dimensional} counterpart to Nielsen's conjecture (with a \emph{topological} notion of ``typical'') was proved by Clifton, Hepburn and W\"{u}thrich~\cite{CHW02}. We note that Nielsen's conjecture does not currently appear to have a direct application to protocols or operational tasks; it should be viewed as a mathematical conjecture about the structure of entanglement transformation.

The main purpose of this paper is to prove Nielsen's conjecture. Due to Nielsen's characterisation of LOCC-convertibility in terms of majorisation, we can actually completely abandon the language of quantum information theory: for a random pure state $|\psi\rangle\in \mb C^{n}\otimes\mb C^{m}$,
the density operator $\rho_{\psi}$ can be interpreted\footnote{Random pure states have been studied for many different purposes throughout the physics literature (with varying levels of rigour), and it is a bit difficult to definitively pin down the origins of this observation. The distribution of the entanglement spectrum of a random bipartite state (which has now become known as the \emph{Hilbert--Schmidt measure}) was perhaps first studied by Lubkin~\cite{Lub78}, and an explicit formula for the density function was first computed by Lloyd and Pagels~\cite{LP88}. The connection to the Wishart--Laguerre spectrum seems to have been first explicitly observed by \.{Z}yczkowski and Sommers~\cite{ZS01}.} as a random matrix sampled from the \emph{trace-normalised complex Wishart--Laguerre ensemble}, defined as follows.
\begin{definition*}
For parameters $n,m\in\mb N$, the \emph{complex Wishart--Laguerre ensemble} is the distribution of the random Hermitian matrix $GG^\dagger$, where ``$\dagger$'' denotes Hermitian transpose, and $G\in \mb C^{n\times m}$ is an $n\times m$ matrix with independent complex standard Gaussian entries\footnote{There are two slightly different conventions for the definition of ``complex standard Gaussian'' (differing by a factor of $\sqrt 2$). Due to the trace-normalisation, the distinction is irrelevant, but for concreteness we use the normalisation that for complex standard Gaussian $X$ we have $\mb E[(\Re X)^2+(\Im X)^2]=1$.}. Moreover, by the \emph{trace-normalised complex Wishart--Laguerre ensemble}, we mean the distribution of $M/\on{tr}(M)$, where $M$ is sampled from the complex Wishart--Laguerre ensemble.
\end{definition*}

So, Nielsen's conjecture is equivalent to the following theorem, which is the main result of this paper.

\begin{theorem}\label{thm:Nielsen-RMT}
Let $A^{(1)},A^{(2)}\in \mb C^{n\times n}$ be independent samples from the trace-normalised complex Wishart--Laguerre ensemble with parameters $n,m$. Let $\lambda_1^{(1)}\ge \dots\ge \lambda_n^{(1)} \ge 0$ and $\lambda_1^{(2)}\ge \dots\ge \lambda_n^{(2)} \ge 0$ be the eigenvalues of $A^{(1)}$ and $A^{(2)}$, respectively. Then
\[\Pr\big[\lambda_k^{(1)}+\dots+\lambda_n^{(1)}\le \lambda_k^{(2)}+\dots+\lambda_n^{(2)}\text{ for all }k\le n\big]\to 0\]
as $n,m\to \infty$.
\end{theorem}
We discuss the proof approach for \cref{thm:Nielsen-RMT} in a bit more detail later in this introduction (\cref{subsec:outline}), but to give a very brief flavour: due to the so-called \emph{eigenvalue repulsion} phenomenon, the (complex) Wishart--Laguerre spectrum has only tiny fluctuations, and existing technology in random matrix theory is not precise enough to describe the distributions of quantities of the form $\lambda_k+\dots+\lambda_n$. So, instead of considering such quantities directly, we consider a sequence of convex test functions which emphasise different parts of the spectrum (increasingly ``focusing in'' on the ``edges'' of the spectrum); this defines a sequence of linear statistics with nearly independent fluctuations, for which we can derive a multivariate central limit theorem. We can then take advantage of the Hardy--Littlewood--P\'olya theorem relating convexity and majorisation, together with some analytic and combinatorial estimates.

\begin{remark*}
Our proof approach, as written, is not quantitative; we prove that the desired majorisation probability converges to zero without proving any explicit bounds on this probability. In principle, it may be possible to establish quantitative central limit theorems for the complex Wishart--Laguerre spectrum, which would provide an explicit bound, though it seems that substantial new ideas would be required to precisely characterise the asymptotic rate of decay of the majorisation probability (computer experiments by Cunden, Facchi, Florio and Gramegna~\cite{CFFG20} indicate that the probability has order of magnitude $n^{-\theta}$ for some constant $\theta>0$, assuming that the ratio $m/n$ is held constant). 
\end{remark*}

\begin{remark*}
We have only discussed LOCC-convertibility of \emph{bipartite} quantum states, and it is natural to also ask about \emph{multipartite} states with more than two subsystems. The theory of LOCC-convertibility becomes much more challenging in this general setting (in the tripartite case, a necessary and sufficient condition is known~\cite{Taj13}, but it is rather complicated). However, our proof of Nielsen's conjecture automatically implies a similar result for any number of subsystems (e.g., if there are three subsystems, we can simply view two of the subsystems as comprising a single larger subsystem to reduce to the bipartite case; this only makes LOCC-convertibility more permissive). Actually, for more than two subsystems one expects much stronger results along these lines; see \cite{GKW17,SWGK18}.
\end{remark*}

\begin{remark*}
It would be interesting to investigate whether similar ideas can be extended to other natural random measures and related partial orders. For fixed-trace orthogonal or symplectic Wishart ensembles, one would need analogues of the multivariate linear-statistics central limit theorems used in \cref{thm:general-CLT,thm:general-CLT-infinite}, together with sufficiently precise control of the corresponding covariance structure after trace-normalisation. For more general measures on the simplex, such as Dirichlet-type laws, one obstacle is that the convenient exact order-statistic representation used in \cref{sec:uniform-measure} is no longer available in the same form. One could also ask for analogues in thermomajorisation, but there the relevant order is governed by thermo-Lorenz curves rather than ordinary majorisation, so the Hardy--Littlewood--P\'olya convex-function reduction used here is not directly available.
\end{remark*}

In addition to \cref{thm:Nielsen-RMT}, we also prove some other related results, which we describe next.

\subsection{The uniform measure on the simplex}
For any possible outcome $M$ of the trace-normalised complex Wishart--Laguerre ensemble, the eigenvalues are nonnegative real numbers that sum to 1. That is to say, the spectrum can be interpreted as a point in the $(n-1)$-simplex
\[\Delta_{n-1}=\{\vec{x}\in[0,\infty)^{n}:x_{1}+\dots+x_{n}=1\},\]
and Nielsen's conjecture can be interpreted as a conjecture about majorisation of two independent random points sampled from a particular probability distribution (defined in random-matrix-theoretic terms) on $\Delta_{n-1}$.

Of course, there are other interesting distributions on $\Delta_{n-1}$; perhaps the most natural (and most tractable) such measure is the \emph{uniform} measure on $\Delta_{n-1}$. Cunden, Facchi, Florio and Gramegna~\cite{CFFG21} proved the natural analogue of Nielsen's conjecture for this measure. Their proof was fundamentally non-quantitative (it critically used Kolmogorov's zero-one law), but based on numerical experiments and heuristic comparison with persistence probabilities of random walks, they predicted that the majorisation probability should decay ``polynomially fast'' (specifically, they predicted that the majorisation probability should scale like $n^{-\theta}$ for some positive constant $\theta\approx 0.41$). We are able to prove a bound of this shape (though presumably not with an optimal value of $\theta$), as follows.

\begin{theorem}\label{thm:uniform}
Let $(X_k^{(1)})_{k=1}^n,(X_k^{(2)})_{k=1}^n$ be two independent uniform random vectors in $\Delta_{n-1}$, and denote by $(\lambda_k^{(1)})_{k=1}^n,(\lambda_k^{(2)})_{k=1}^n$ their \emph{decreasing rearrangements} (i.e., the results of sorting $(X_k^{(1)})_{k=1}^n,(X_k^{(2)})_{k=1}^n$ in decreasing order). Then
\[\Pr\big[\lambda_k^{(1)}+\dots+\lambda_n^{(1)}\le \lambda_k^{(2)}+\dots+\lambda_n^{(2)}\text{ for all }k\le n\big]= O((n/\log n)^{-1/16}).\]
\end{theorem}

As observed by Cunden, Facchi, Florio and Gramegna~\cite{CFFG21}, this result actually has a quantum-information-theoretic interpretation: in the same way that Nielsen's conjecture studies the resource of \emph{entanglement} (i.e., the operations that can be made without introducing new entanglement into the system), \cref{thm:uniform} can be interpreted as a theorem about the resource of \emph{coherence}. We direct the reader to \cite{CFFG21}, and the references therein, for more details.

\subsection{Approximate LOCC-convertibility} So far, we have only been concerned about \emph{exact} LOCC-convertibility: we write $|\psi\rangle\to |\phi\rangle$ if it is possible for Alice and Bob to orchestrate a sequence of local quantum operations (with unlimited classical communication) which \emph{guarantee}, with probability 1, that $|\psi\rangle$ will be converted to $|\phi\rangle$. It is natural to ask how the situation changes if we allow some probability of failure.

Shortly after Nielsen's 1999 paper, Vidal~\cite{Vid99} found a generalisation of Nielsen's theorem along these lines. For bipartite states $|\psi\rangle,|\phi\rangle\in \mb C^{n}\otimes\mb C^{m}$, he proved that the maximum possible success probability $\Pi(|\psi\rangle\to |\phi\rangle)$ for the task of converting $|\psi\rangle$ to $|\phi\rangle$ (in the LOCC paradigm) is exactly
\[\min_{1\le k\le n} \frac{\lambda_k^{(1)}+\lambda_{k+1}^{(1)}+\dots+\lambda_n^{(1)}}{\lambda_k^{(2)}+\lambda_{k+1}^{(2)}+\dots+\lambda_n^{(2)}},\]
where $\lambda_1^{(1)}\ge \dots\ge \lambda_n^{(1)}$ are the eigenvalues of $\rho_{\psi}$ and $\lambda_1^{(2)}\ge \dots\ge \lambda_n^{(2)}$ are the eigenvalues of $\rho_{\phi}$.

In \cite{CFFG20}, Cunden, Facchi, Florio and Gramegna considered $\Pi(|\psi\rangle\to |\phi\rangle)$ for \emph{random} states $|\psi\rangle,|\phi\rangle\in \mb C^{n}\otimes\mb C^{m}$: with what probability can one successfully convert between two randomly entangled states? In their computer simulations, they found that the answer seems to heavily depend (in a somewhat surprising way) on the relationship between $n$ and $m$. Indeed, suppose $n,m\to \infty$ with $m/n=c$ (for some constant $c\in (0,\infty)$). If $c=1$, computational evidence suggests that $\Pi(|\psi\rangle\to |\phi\rangle)$ has a nontrivial limiting distribution supported on the whole interval $[0,1]$ (with mean about 0.58). On the other hand, if $c\ne 1$ then computational evidence suggests that $\Pi(|\psi\rangle\to |\phi\rangle)$ converges to 1 in probability.

As our final result in this paper, we confirm the latter prediction.

\begin{theorem}\label{thm:approximate-conversion}
Let $A^{(1)},A^{(2)}\in \mb C^{n\times n}$ be independent samples from the trace-normalised complex Wishart--Laguerre ensemble with parameters $n,m$. Let $\lambda_1^{(1)}\ge \dots\ge \lambda_n^{(1)}$ and $\lambda_1^{(2)}\ge \dots\ge \lambda_n^{(2)}$ denote the respective eigenvalues, and let
    \[\Pi=\min_{1\le k\le n} \frac{\lambda_k^{(1)}+\lambda_{k+1}^{(1)}+\dots+\lambda_n^{(1)}}{\lambda_k^{(2)}+\lambda_{k+1}^{(2)}+\dots+\lambda_n^{(2)}}.\]
Then, for any fixed $\varepsilon>0$ and $c\ne 1$, we have
\[\mb{P}[\Pi < 1 - \eps] \to 0\]
as $n,m\to \infty$ with $m/n\to c$.
\end{theorem}
\begin{remark}
\label{rmk:approx-locc-equality-case}
The conclusion of \cref{thm:approximate-conversion} cannot hold without some separation between $m$ and $n$. Indeed, in the case $m = n$, Edelman \cite{edelman1988eigenvalues} showed that the smallest eigenvalue $\lambda_{n}$ of a matrix drawn from the complex Wishart--Laguerre ensemble has probability density function given by $f(\lambda) = ne^{-\lambda n/2}/2$. Using the concentration of the trace (as in the proof of \cref{thm:approximate-conversion}) along with the observation that $\Pi \leq \lambda_n^{(1)}/\lambda_{n}^{(2)}$, we see that for every $\eps < 1$, there exists $c_{\eps} > 0$ such that $\mb{P}[\Pi < 1-\eps] > c_{\eps}$.

Our proof of \cref{thm:approximate-conversion} really shows that $\mb{P}[\Pi < 1 - \eps] \to 0$ as long as $|m-n|$ grows faster than $\sqrt{n \log n}$ (see \cref{thm:approximate-conversion-quantitative} for a precise statement), and it is conceivable that existing results in random matrix theory can be used to extend the above argument (showing $\mb{P}[\Pi < 1-\eps] > c_{\eps}$) to the range $|m - n| = O(\sqrt{n})$. However, this would still leave a ``gap'' of order $\sqrt{\log n}$. It would be interesting to investigate this further.
\end{remark}

\subsection{Discussion of proof ideas}\label{subsec:outline}
If $(\lambda_k^{(1)})_{k=1}^n$ and $(\lambda_k^{(2)})_{k=1}^n$ are the decreasing rearrangements of two independent random vectors on the simplex, sampled according to any continuous distribution, then for every \emph{individual} $k$, it is easy to see by symmetry considerations that $\lambda_k^{(1)}+\dots+\lambda_n^{(1)}\le \lambda_k^{(2)}+\dots+\lambda_n^{(2)}$ with probability exactly 1/2. So, a first instinct for how to study majorisation of $(\lambda_k^{(1)})_{k=1}^n$ and $(\lambda_k^{(2)})_{k=1}^n$ (e.g., in the settings of \cref{thm:Nielsen-RMT,thm:uniform}) is to study the dependence between the events that $\lambda_k^{(1)}+\dots+\lambda_{n}^{(1)}\le \lambda_k^{(2)}+\dots+\lambda_{n}^{(2)}$ for different $k$. For example, if we consider a sequence of indices $k_1<\dots<k_\ell$ which grow very rapidly or are in some other sense ``well-spaced'' from each other, one might hope that the corresponding events are approximately independent from each other, meaning that one might hope for an upper bound of about $2^{-\ell}$ on the majorisation probability. This type of approach has been very successful for some similar problems about majorisation in random settings (see for example \cite{Pit18,Pit99,MMM21} for its use in the study of random integer partitions).

This is the approach we take in our proof of \cref{thm:uniform}. Indeed, we employ an \emph{exact description} of the joint distribution of $(\lambda_k)_{k=1}^n$, in terms of independent exponential random variables. Via some simple approximation arguments, we can relate the quantities $\lambda_1^{(1)}+\dots+\lambda_{k}^{(1)}$ to so-called \emph{iterated partial sums} (also called \emph{integrated random walks}). There is a fairly extensive literature on persistence of iterated partial sums (see for example \cite{AD13,dembo2013persistence,Vys10,Sin92}); in particular we take advantage of a general estimate due to Dembo, Ding and Gao~\cite{dembo2013persistence}. 

Unfortunately, the above strategy seems to be extremely difficult to apply in the setting of \cref{thm:Nielsen-RMT}. In this setting there does not seem to be any convenient exact description of the distribution of $\lambda_k+\dots+\lambda_n$, so it seems we are forced to turn to limit theorems. There is a huge body of research on central limit theorems for certain statistics of Wishart--Laguerre ensembles (see for example \cite{Wish1,Wish2,wish3,Wish4,Arh71,Jon82,lytova2009central,Bao15,CP15}), but none of this applies to quantities of the form $\lambda_k+\dots+\lambda_n$. The difficulty here is that there is a so-called \emph{eigenvalue repulsion} phenomenon for Hermitian random matrices: the eigenvalues ``want to be separated from each other'', which has the effect of ``locking the spectrum in place''. Specifically (at least in the case where $m$ and $n$ are approximately equal), for any bounded test function $f$, the fluctuations of $f(\lambda_1)+\dots+f(\lambda_n)$ are tiny, of comparable size to the eigenvalues themselves (and so there is no chance of describing the fluctuations in $\lambda_k+\dots+\lambda_n$ via a central limit theorem for linear statistics of the form $f(\lambda_1)+\dots+f(\lambda_n)$).

Instead, we take advantage of the relationship between majorisation and convex functions: the \emph{Hardy--Littlewood--P\'olya theorem} says that if $(\lambda_k^{(1)})_{k=1}^n$ majorises $(\lambda_k^{(2)})_{k=1}^n$, then for any convex function $f:\mb R \to \mb R$ we have \[f(\lambda_1^{(1)})+\dots+f(\lambda_n^{(1)})\ge f(\lambda_1^{(2)})+\dots+f(\lambda_n^{(2)}).\]
So, we choose a sequence of convex test functions $f_1,\dots,f_\ell$ which ``emphasise distinct parts of the spectrum'' and whose corresponding statistics therefore fluctuate almost independently. Specifically, for each $i$, we choose the function $f_i$ to grow much more rapidly than $f_1,\dots,f_{i-1}$, and therefore focus much more strongly on the upper edge (the so-called \emph{soft edge}) of the spectrum.  For simplicity, we will ultimately take the functions $f_j$ to be polynomials of rapidly growing degree.
We can then apply a general multivariate central limit theorem (with an adjustment for trace-normalisation) to the $\ell$ different statistics of the form $f_i(\lambda_1)+\dots+f_i(\lambda_n)$, and via analysis of certain integrals we can upper-bound the covariances between these statistics. The upshot is that we are able to show that the majorisation probability is at most about $2^{-\ell}$; since we can take $\ell$ arbitrarily large, the desired result follows.

The proof of \cref{thm:approximate-conversion} is completely different. Since the trace is extremely well-concentrated, it suffices to prove a version of \cref{thm:approximate-conversion} for eigenvalues of the \emph{non-trace-normalised} complex Wishart--Laguerre ensemble. Writing $\mu_1,\dots,\mu_n$ for these non-normalised eigenvalues, our approach is to prove that each individual eigenvalue $\mu_k$ is tightly concentrated (in a multiplicative sense) around its expected value $\mb E\mu_k$.
Our proof of this fact relies on two ingredients: first, by Weyl's inequalities, $\mu_k^{1/2}$ can be interpreted as a $1$-Lipschitz function (with respect to the Euclidean norm) of a Gaussian random vector, so it follows from the Gaussian concentration inequality of Sudakov--Tsirelson and Borell that $\mu_k$ is well-concentrated in an \emph{additive} sense. Second, since $m \geq n + C\sqrt{n\log n}$, it follows from known estimates on the smallest singular value of random matrices that $\mb{E}[\mu_{n}^{1/2}]$ is sufficiently large to convert the additive error to multiplicative error. The second ingredient is the only place where we use the separation between $m$ and $n$; as we discuss in \cref{rmk:approx-locc-equality-case}, this is necessary.  

\section{Proof of Nielsen's conjecture}
In this section we present the proof of \cref{thm:Nielsen-RMT} (Nielsen's conjecture), modulo some analytic estimates that will be deferred to \cref{sec:integral-estimates}.

\cref{thm:Nielsen-RMT} concerns parameters $m,n$ which both tend to infinity, without making any assumption about the relationship between $m$ and $n$. However, it is straightforward to reduce to the case where $m\ge n$ (by the symmetry between $m$ and $n$), and it is then straightforward to reduce to the case where $m/n\to c$ for some constant $c\in [1,\infty]$ (by compactness). There are certain complications in the $c=\infty$ (``imbalanced'') case that are not present in the $c<\infty$ (``approximately balanced'') case, so we choose to present these two cases separately.

\subsection{Majorisation and convexity}
First, as discussed in \cref{subsec:outline}, it is crucial that instead of working with majorisation directly, we can work with convex test functions. To this end, we need the \emph{Hardy--Littlewood--P\'olya} theorem (appearing, for example, as \cite[p.~259, Theorem~B]{RV73}).
\begin{theorem}\label{thm:convex-majorisation}
    Consider any real numbers $\lambda_1^{(1)}\ge \dots\ge \lambda_n^{(1)}$ and $\lambda_1^{(2)}\ge \dots\ge \lambda_n^{(2)}$, such that
    \[\lambda_k^{(1)}+ \dots+ \lambda_n^{(1)}\le\lambda_k^{(2)}+ \dots+ \lambda_n^{(2)}\text{ for all }k\in \{1,\dots,n-1\}\quad \text{and}\quad\lambda_1^{(1)}+ \dots+ \lambda_n^{(1)}=\lambda_1^{(2)}+ \dots+ \lambda_n^{(2)}.\]
    Then for any convex function $f:\mb R\to \mb R$, we have
    \[f(\lambda_1^{(1)})+\dots+f(\lambda_n^{(1)})\le f(\lambda_1^{(2)})+\dots+f(\lambda_n^{(2)}).\]
\end{theorem}

In particular, to show that a sequence $(\lambda_k^{(1)})_{k=1}^{n}$ is \emph{not} majorised by $(\lambda_k^{(2)})_{k=1}^{n}$, it suffices to exhibit a convex function $f$ for which $\sum_{k=1}^{n} f(\lambda_k^{(1)}) > \sum_{k=1}^{n} f(\lambda_k^{(2)})$. To prove \cref{thm:Nielsen-RMT}, we will consider an explicit sequence of convex test functions $f_i$, and show that the events of the form $\sum_{k=1}^{n} f_i(\lambda_k^{(1)}) > \sum_{k=1}^{n} f_i(\lambda_k^{(2)})$ (for different $i$) are nearly independent from each other, meaning that it is very likely that at least one of these events holds (and therefore $(\lambda_k^{(1)})_{k=1}^{n}$ is likely not majorised by $(\lambda_k^{(2)})_{k=1}^{n}$).

\subsection{A central limit theorem for linear statistics, in the approximately balanced case}
The other key general ingredient we will need is a (joint) central limit theorem for linear statistics of the complex Wishart--Laguerre spectrum. To start with, we present a central limit theorem that is suitable for the $c<\infty$ case (i.e., the case where $m$ and $n$ are approximately balanced).
The statement of our central limit theorem requires some preparation.

\begin{definition}\label{def:CLT-prep}
    For continuous functions $f,g:\mb R\to \mb R$, define
    \begin{align*}\Gamma(f,g)&=\frac{1}{\pi^2} \int_{-1}^{1}\int_{-1}^{1} \left(\frac{f(x) - f(y)}{x-y} \right) \left(\frac{g(x) - g(y)}{x-y} \right) \frac{1 - xy}{\sqrt{1 - x^2} \sqrt{1 - y^2}}\,dx\,dy.\end{align*}
    Then, for $c \in [1,\infty)$, let $a_{\pm} = (1 \pm \sqrt{c})^2$ and $\Phi_c(x) = 2\sqrt{c} x + c + 1$, and for continuous $f,g:\mb R\to \mb R$, define
    \[\gamma_c(f) = \frac{1}{2\pi} \int_{a_-}^{a_+} f(x)\frac{\sqrt{(a_+ - x)(x - a_-)}}{x} \,dx,\qquad 
    \Gamma_c(f,g) = \Gamma(f\circ \Phi_c,g\circ \Phi_c).\]
    Also, let $\mc F$ be the set of all continuously differentiable functions $f:\mb R\to \mb R$ with ``sub-exponential growth'', in the sense that $(\log |f(x)|)/|x|\to 0$ when $|x|\to \infty$. 
\end{definition}
\begin{theorem}\label{thm:general-CLT}
    Fix functions $f_1,\dots,f_\ell\in \mc F$ and let $\mu_1,\dots,\mu_n$ be the eigenvalues of a random matrix sampled from the complex Wishart--Laguerre ensemble with parameters $m, n$.  Suppose $m,n \to\infty$ in such a way that $m/n \to c \in [1,\infty)$.      
    For each $i\in\{1,\dots,\ell\}$ let
\[X_i= \sum_{k = 1}^n f_i(\mu_k/n).\]
Then the following two asymptotic properties hold.
\begin{enumerate}
    \item \label{item:LLN} For each $i\in\{1,\dots,\ell\}$, we have the $L^1$-convergence \[\E \big|X_i/n- \gamma_c(f_i)\big|\to 0.\]
    \item \label{item:CLT} We have the multivariate convergence in distribution\[\big(X_1-\mb EX_1,\;\dots,\;X_\ell-\mb EX_\ell)\overset d\to \mc N(\vec 0,\Sigma),\] where the limiting covariance matrix $\Sigma\in \mb R^{\ell\times\ell}$ has $(i,j)$-entry equal to $\Gamma_c(f_i,f_j)$.
\end{enumerate}
\end{theorem}

\cref{thm:general-CLT} can be proved by combining a few different results in the literature. The asymptotic formulas for the means in (\cref{item:LLN}) go back to the seminal work of Marchenko and Pastur~\cite{MP67}, and a central limit theorem was first proved by Arharov~\cite{Arh71}. However, Arharov
did not provide formulas for the covariances; for the covariance formulas in (\cref{item:CLT}), we use a result of Lytova and Pastur~\cite{lytova2009central}.
There are three differences between these results and \cref{thm:general-CLT}: first, our result is valid for a larger class of test functions; second, in (\cref{item:LLN}), we obtain convergence in mean (whereas previous results obtain almost-sure convergence or convergence in probability---the distinction will be important for our application), and third, our result is a multivariate central limit theorem (whereas most previous results are univariate). In \cref{sec:CLT} we show how to deduce \cref{thm:general-CLT} from previous results (the deduction is standard, but we include this for completeness).

\begin{remark*}
    It is perhaps more common to study the \emph{real} Wishart--Laguerre ensemble, which is the distribution of $A A^\transpose$ when $A\in \mb R^{n\times m}$ is a random matrix with independent \emph{real} standard Gaussian entries. The distinction between the real and complex case does not affect \cref{thm:general-CLT}\eqref{item:LLN}, and affects the limiting covariances in \cref{thm:general-CLT}\eqref{item:CLT} by a factor of exactly 2 (i.e., the covariances are twice as large in the real case as in the complex case); see e.g.\ \cite[Remark~4.1]{lytova2009central}.
\end{remark*}

In our proof of \cref{thm:Nielsen-RMT}, the only reason we need \cref{thm:general-CLT} is for the following corollary, providing a central limit theorem for sums of powers of trace-normalised eigenvalues.

\begin{corollary}\label{cor:c<infty}
    Let $\mu_1,\dots,\mu_n$ be the eigenvalues of a random matrix sampled from the complex Wishart--Laguerre ensemble with parameters $m\geq n$ and let $\lambda_1,\ldots,\lambda_n$ denote the trace-normalised eigenvalues. Fix $\ell$, and for each $i\in\{1,\dots,\ell\}$ let
    \[X_i= \sum_{k = 1}^n (\mu_k/n)^i \quad \text{ and } \quad Y_i = \sum_{k = 1}^n \lambda_k^i.\]
    Suppose $n,m\to \infty$ in such a way that $m/n\to c$, for some fixed $c\in [1,\infty)$.  Then we have multivariate convergence in distribution $$\Bigg( n\bigg(Y_i \cdot \frac{(\E X_1)^i}{\E X_i} - 1\bigg) \Bigg)_{i \leq \ell} \overset d \to  \left(\frac{Z_i}{\gamma_c(x^i)} - i \frac{Z_1}{\gamma_c(x)}\right)_{i \leq \ell}$$
   where $(Z_i)_{i \leq \ell} \sim \mathcal{N}(\vec 0, \Sigma)$ for a matrix
$\Sigma\in \mb R^{\ell\times\ell}$ with $(i,j)$-entry $\Gamma_c(x^i,x^j)$.
\end{corollary}
(Here we are abusing notation slightly, writing $x^i$ to indicate the function $x\mapsto x^i$).

\begin{proof}
    Writing $\Zt_i = X_i - \E X_i$, note that by \cref{thm:general-CLT} we have convergence in distribution $(\Zt_i)_{i \leq \ell} \to (Z_i)_{i \leq \ell}.$  We may then expand \begin{align}
    \label{eq:Y-expand}
    Y_i \cdot \frac{(\E X_1)^i}{\E X_i} = \frac{X_i}{X_1^i} \cdot \frac{(\E X_1)^i}{\E X_i}  = \left(1 + \frac{\Zt_i}{\E X_i} \right)\left( 1 + \frac{\Zt_1}{\E X_1}\right)^{-i}  = 1 + \frac{\Zt_i}{\E X_i} - \frac{i \Zt_1}{\E X_1} + O(n^{-2}).
    \end{align}
    Here, asymptotics are ``in probability'', treating $c,i$ as constants: the notation $O(n^{-2})$ denotes a random variable $E_{c,i}$ such that $n^2 E_{c,i}$ is bounded in probability (as $n\to \infty$, holding $c,i$ fixed). This estimate for $E_{c,i}$ follows from the fact that the random variables $\Zt_i$ are bounded in probability (which is a consequence of \cref{thm:general-CLT}\eqref{item:CLT}) and the fact that $\E X_i = n\gamma_{c}(x^i) + o(n)$ has order of magnitude $n$ (here we used \cref{thm:general-CLT}\eqref{item:LLN}, and the observation that $\gamma_c(x^i) \in (0,\infty)$ for all $c \in [1,\infty)$ and $i \geq 1$).

    Rearranging \eqref{eq:Y-expand} and using again that $\E X_i = n\gamma_c(x^i) + o(n)$ completes the proof.
\end{proof}

\begin{remark*}
Our choice of test functions $x\mapsto x^i$ is mainly to keep the proof of \cref{cor:c<infty} (and its counterpart for the imbalanced case, \cref{cor:c=infty}) as simple as possible. If one were interested in optimising the quantitative aspects of our proof strategy, one should presumably consider alternative test functions.
\end{remark*}

\subsection{The covariance structure, and putting the pieces together} \label{subsec:c<inf} In order to use \cref{cor:c<infty}, we will need some control over the quantities $\gamma_c$ and $\Gamma_c$, which describe the limiting means and covariances of our test functions. This is the content of the next lemma, whose proof we defer to \cref{sec:integral-estimates}. 
Recall the definition of $a_+$ from \cref{def:CLT-prep}.

\begin{lemma}\label{lem:technical}
    Fix $c \in [1,\infty)$. For $i \in \mathbb{N}$, let $h_i(x) = (x/a_+)^i$.
    \begin{enumerate}
        \item We have $\lim_{i \to \infty}  i\gamma_c( h_i ) = 0$. \label{item:gamma}
        \item The limit $\alpha_c=\lim_{i \to \infty} \Gamma_c(h_i,h_i)$ exists, and satisfies $\alpha_c>0$.
        \label{item:cov-diagonal}
        \item We have $\lim_{A \to \infty} \lim_{i \to \infty} \Gamma_c(h_i, h_{Ai}) = 0$.
        \label{item:cov-off-diagonal}
    \end{enumerate}    
\end{lemma}
Recall the random variables $Y_i=\sum_{k=1}^n \lambda_k^i$ from \cref{cor:c<infty}, which can be interpreted as linear statistics associated with convex test functions $f_i:x\mapsto x^i$. The estimates in \cref{lem:technical} can be interpreted as saying that the fluctuations of $Y_i$ and $Y_j$ are nearly independent, as long as $i$ is very large and $j$ is much larger than $i$ (specifically, \eqref{item:gamma} can be used to show that correlations due to trace-normalisation are not very impactful, and \eqref{item:cov-diagonal} and \eqref{item:cov-off-diagonal} can then be used together to show that the covariance between the non-trace-normalised random variables $X_i$ and $X_j$ is negligible compared to the variances of $X_i$ and $X_j$ individually).

\begin{remark*}
One can compute the relevant integrals in \cref{lem:technical}, and apply them to prove \cref{thm:Nielsen-RMT}, without any intuitive understanding of \emph{why} estimates of this type should hold. However, roughly speaking, the picture to keep in mind is that if $j$ is much larger than $i$, then the largest eigenvalues (and their fluctuations) contribute much more strongly to $Y_j$ than to $Y_i$. That is to say, $Y_j$ is dominated by a few very large eigenvalues, which play only a negligible role in $Y_i$. So, intuitively speaking, \cref{lem:technical} corresponds to the fact that different parts of the Wishart--Laguerre spectrum have nearly independent fluctuations, and the sizes of these fluctuations do not significantly decay as we approach the upper edge of the spectrum (the latter is important as there is a ``competition'' between these fluctuations and the fluctuations due to trace-normalisation). We remark that the \emph{lower} edge of the spectrum does not enjoy these properties: since all eigenvalues are always at least zero, the fluctuations of the smallest eigenvalues are more constrained. The upper and lower edges of the spectrum are sometimes called the ``soft edge'' and ``hard edge'' for this reason.
\end{remark*}

Given the preparations in this section so far, we can now prove Nielsen's conjecture in the case where $m/n\to c$ for $c\in [1,\infty)$.

\begin{proof}[Proof of \cref{thm:Nielsen-RMT}, in the case $m/n\to c\in [1,\infty)$] 
    
    Let $\mathcal E_{\mr{maj}}$ be the event that $\lambda_k^{(1)}+\dots+\lambda_n^{(1)}\le \lambda_k^{(2)}+\dots+\lambda_n^{(2)}$ for all $k \leq n$. Recall that we wish to show that $\mb{P}[\mc{E}_{\mr{maj}}] \to 0$, when $m,n \to \infty$ in such a way that $m/n\to c\in [1,\infty)$.
    
    For $s \in \{1,2\}$ define $Y_i^{(s)} = \sum_{j = 1}^n (\lambda_j^{(s)})^i.$ 
    By convexity of the functions $x \mapsto x^i$ on $[0,\infty)$, and \cref{thm:convex-majorisation}, we see that on the event $\mathcal{E}_{\mathrm{maj}}$ 
    we must have $Y_i^{(1)} \leq Y_i^{(2)}$.  By \cref{cor:c<infty} this implies that for any $\ell$, we have $$\limsup_{n \to \infty} \P[\mathcal{E}_{\mathrm{maj}}] \leq \P\left[Z_i^{(1)} \leq Z_i^{(2)} - \frac{i \gamma_c(x^i)}{\gamma_c(x)}\big(Z_1^{(2)} - Z_1^{(1)} \big) \, \text { for all } i  \leq \ell \right],$$
    where $(Z_i^{(1)})_{i \leq \ell}$ and $(Z_i^{(2)})_{i \leq \ell}$ are independent Gaussian vectors as defined in \cref{cor:c<infty}.

    Since the left-hand side is independent of $\ell$, we may take $\ell\to\infty$. Writing
\[
\mc E_\ell=\left\{Z_i^{(1)} \leq Z_i^{(2)} - \frac{i \gamma_c(x^i)}{\gamma_c(x)}\bigl(Z_1^{(2)} - Z_1^{(1)} \bigr)\, \text{for all } i \le \ell \right\},
\]
note that the events $\mc E_\ell$ are decreasing in $\ell$. Therefore, by continuity from above,
\begin{align*}
\limsup_{n \to \infty} \P[\mathcal{E}_{\mathrm{maj}}]
\le \lim_{\ell \to \infty}\P[\mc E_\ell] = \mb{P}\left[Z_i^{(1)} \leq Z_i^{(2)} - \frac{i \gamma_c(x^i)}{\gamma_c(x)}\bigl(Z_1^{(2)} - Z_1^{(1)} \bigr)\, \text{for all } i\right].
\end{align*}

Define $W_i^{(s)} = Z_i^{(s)} / a_+^i$, and let $h_i(x)=(x/a_+)^i$ as in \cref{lem:technical}. Note that the covariance between $W_i^{(s)}$ and $W_j^{(s)}$ is precisely $\Gamma_c(h_i,h_j)$, and note that
\begin{align*}
\P\left[Z_i^{(1)} \leq Z_i^{(2)} - \frac{i \gamma_c(x^i)}{\gamma_c(x)}\bigl(Z_1^{(2)} - Z_1^{(1)} \bigr)\, \text{for all } i\right]
= \P\left[W_i^{(1)} \leq W_i^{(2)} - \frac{i \gamma_c(h_i)}{\gamma_c(x)}\bigl(Z_1^{(2)} - Z_1^{(1)}\bigr)\, \text{for all } i \right].
\end{align*}

Now, for each $j,q,A \in \mathbb{N}$, set $i_j(q,A) = q A^j$. So, by \cref{lem:technical}, as we send $A,q\to \infty$ (where $A\to \infty$ much more rapidly than $q\to \infty$), we observe:
\begin{enumerate}
    \item $i_j(q,A) \gamma_c(h_{i_j(q,A)})\to 0$;
    \item the variances of the individual $W_{i_j(q,A)}^{(s)}$ tend to $\alpha_c>0$;
    \item the covariances between different $W_{i_j(q,A)}^{(s)}$ tend to zero.
\end{enumerate}

Fix $N\in \mb N$. For each $q,A$, the random vector
\[
\Bigl(W_{i_j(q,A)}^{(1)},W_{i_j(q,A)}^{(2)}\Bigr)_{j\le N}
\]
is centred Gaussian. By the three observations above, its covariance matrix converges entrywise, as first $q\to\infty$ and then $A\to\infty$, to the diagonal matrix $\alpha_c I_{2N}$. Since a centred Gaussian law is determined by its covariance matrix, it follows that
\[
\Bigl(W_{i_j(q,A)}^{(1)},W_{i_j(q,A)}^{(2)}\Bigr)_{j\le N}
\overset d\to
\Bigl(\zeta_j^{(1)},\zeta_j^{(2)}\Bigr)_{j\le N},
\]
where $(\zeta_j^{(1)},\zeta_j^{(2)})_{j\le N}$ are i.i.d.\ real Gaussians with mean zero and variance $\alpha_c$.

Also, if we define
\[
b_{j,q,A}=\frac{i_j(q,A)\gamma_c(h_{i_j(q,A)})}{\gamma_c(x)},
\]
then by (\cref{item:gamma}) we have $\max_{j\le N}|b_{j,q,A}|\to 0$ as first $q\to\infty$ and then $A\to\infty$. Since $Z_1^{(2)} - Z_1^{(1)}$ is bounded in probability, it follows that
\[
\Bigl(b_{j,q,A}\bigl(Z_1^{(2)} - Z_1^{(1)}\bigr)\Bigr)_{j\le N}\to \vec 0
\qquad\text{in probability.}
\]
Therefore, by Slutsky's theorem,
\[
\Bigl(W_{i_j(q,A)}^{(2)} - W_{i_j(q,A)}^{(1)} - b_{j,q,A}\bigl(Z_1^{(2)} - Z_1^{(1)}\bigr)\Bigr)_{j\le N}
\overset d\to
\Bigl(\zeta_j^{(2)} - \zeta_j^{(1)}\Bigr)_{j\le N}.
\]
Since the limiting law is absolutely continuous, the boundary of the orthant event has probability zero, and hence
\begin{align*}
\P&\left[W_i^{(1)} \leq W_i^{(2)} - \frac{i \gamma_c(h_i)}{\gamma_c(x)}\bigl(Z_1^{(2)} - Z_1^{(1)}\bigr)\, \text{for all } i \right] \\
&\le \lim_{A \to \infty}\lim_{q \to \infty} \P\left[W_{i_j(q,A)}^{(1)} \leq W_{i_j(q,A)}^{(2)} - \frac{i_j(q,A) \gamma_c(h_{i_j(q,A)})}{\gamma_c(x)}\bigl(Z_1^{(2)} - Z_1^{(1)}\bigr)\, \text{for all } j \le N \right] \\
&= \P\big[\zeta_j^{(1)} \leq \zeta_j^{(2)}\, \text{for all } j \le N\big] = 2^{-N}.
\end{align*}
Taking $N\to\infty$ completes the proof.
\end{proof}

\subsection{The imbalanced case} In this subsection, we present analogues of \cref{thm:general-CLT}, \cref{cor:c<infty}, and \cref{lem:technical} for the ``imbalanced'' case where $m/n\to \infty$. We then deduce the statement of \cref{thm:Nielsen-RMT} in this case, and provide the details for how to deduce the full statement of \cref{thm:Nielsen-RMT} from the approximately-balanced and imbalanced cases.

The main technical difficulty in the imbalanced case is that the spectrum of the complex Wishart--Laguerre ensemble is concentrated around $m$ (i.e., typically all eigenvalues are of the form $m-o(m)$). This means that linear statistics $Y_i=\sum_{k=1}^{n} \lambda_k^i$ of the type considered in \cref{cor:c<infty} have variance $o(1)$, and their limiting distribution is degenerate.

The root of the problem is that the normalisation in \cref{thm:general-CLT} (considering statistics of the form $\sum_{k=1}^nf(\mu_k/n)$) is not appropriate in the imbalanced regime, and to get a sensible central limit theorem, one should renormalise with an ``additive shift'': namely, one should consider statistics of the form $\sum_{k=1}^nf\big((\mu_k-m)/\sqrt{nm}\big)$. The following central limit theorem is an analogue of \cref{thm:general-CLT} with this renormalisation.
\begin{definition}\label{def:CLT-prep-2}
    Recall the set of functions $\mc F$ with ``sub-exponential growth'', and the notation $\Gamma(f,g)$, from \cref{def:CLT-prep}.    
    For a continuous function $f:\mb R\to \mb R$, let
    $$\gamma(f) = \frac{1}{\pi}\int_{-1}^1 f(x) \sqrt{1 - x^2}\,dx.$$
\end{definition}

\begin{theorem}\label{thm:general-CLT-infinite}
Fix functions $f_1,\dots,f_\ell\in \mc F$ and let $\mu_1,\dots,\mu_n$ be the eigenvalues of a random matrix sampled from the complex Wishart--Laguerre ensemble with parameters $m, n$.  Suppose $m,n \to\infty$ in such a way that $m/n \to \infty$. For each $i\in \{1,\dots,\ell\}$ let
\[X_i= \sum_{k = 1}^n f_i\left(\frac{\mu_k - m}{2\sqrt{nm}}\right).\]
Then the following two asymptotic properties hold.
\begin{enumerate}
    \item \label{item:LLN-infinite} For each $i\in\{1,\dots,\ell\}$, we have the $L^1$-convergence \[\E \big|X_i/n- \gamma(f_i)\big|\to 0.\]
    \item \label{item:CLT-infinite} We have the multivariate convergence in distribution\[\big(X_1-\mb EX_1,\;\dots,\;X_\ell-\mb EX_\ell)\overset d\to \mc N(\vec 0,\Sigma),\] where $\Sigma\in \mb R^{\ell\times\ell}$ is a matrix with $(i,j)$-entry $\Gamma(f_i,f_j)$.
\end{enumerate}
\end{theorem}



As for \cref{thm:general-CLT}, we can prove \cref{thm:general-CLT-infinite} by combining various results in the literature (the deduction can again be found in \cref{sec:CLT}). The two normalisations used in \cref{thm:general-CLT,thm:general-CLT-infinite} correspond to the two standard asymptotic spectral regimes for Wishart matrices. When $m/n\to c<\infty$, the empirical spectral measure of $(\mu_k/n)_{k=1}^n$ converges to the Marchenko--Pastur law with parameter $c$, which underlies \cref{thm:general-CLT}. When $m/n\to\infty$, after centering at $m$ and scaling by $2\sqrt{mn}$, the empirical spectral measure converges to the semicircular law, which underlies \cref{thm:general-CLT-infinite}. In this sense, \cref{thm:general-CLT-infinite} is the $c\to\infty$ counterpart of \cref{thm:general-CLT}. Specifically, the formulas for the means in \cref{thm:general-CLT-infinite}\eqref{item:LLN-infinite} are essentially due to Bai and Yin~\cite{BY88}, and a central limit theorem in this regime was proved independently by Bao~\cite{Bao15} and by Chen and Pan~\cite{CP15}. See also Nechita~\cite{Nech07} for related results on random density matrices.

\begin{remark*}
    We have stated \cref{thm:general-CLT-infinite} only for the case $m/n\to \infty$ (as we have already handled the case where $m/n$ converges to a finite limit), but we remark that it would be possible to state a more general (albeit somewhat complicated) central limit theorem that holds for all $m,n\to \infty$, and use this to give a unified proof of \cref{thm:Nielsen-RMT} without a case distinction.
\end{remark*}

Now, as an analogue of \cref{cor:c<infty}, we next deduce a corollary of \cref{thm:general-CLT-infinite} for powers of (suitably shifted) trace-normalised eigenvalues. Due to the shifting, the deduction is somewhat more complicated than for \cref{cor:c<infty}.

\begin{corollary}\label{cor:c=infty}
    Let $\mu_1,\dots,\mu_n$ be the eigenvalues of a random matrix sampled from the complex Wishart--Laguerre ensemble with parameters $n,m$ and let $\lambda_1,\ldots,\lambda_n$ denote the trace-normalised eigenvalues. Fix $\ell$, and for each $i\in\{1,\dots,2\ell\}$ let
    \[X_i= \sum_{k = 1}^n \left(\frac{\mu_k -m}{2\sqrt{mn}}\right)^i \quad \text{ and } \quad Y_i = \sum_{k = 1}^n \left(\lambda_k - \frac{1}{n}\right)^i.\]
        Suppose $n,m\to \infty$ in such a way that $m/n\to \infty$. Then we have multivariate convergence in distribution $$\Bigg( n\bigg(Y_{2i} \cdot \frac{(mn/4)^{i}}{(\E X_{2i})} - 1\bigg) \Bigg)_{i \leq \ell} \overset d\to  \left(\frac{Z_{2i}}{\gamma(x^{2i})}\right)_{i \leq \ell}$$
   where $(Z_{2i})_{i \leq \ell} \sim \mathcal{N}(\vec 0, \Sigma)$ for a matrix
$\Sigma\in \mb R^{\ell\times\ell}$ with $(i,j)$-entry $\Gamma(x^{2i},x^{2j})$.
\end{corollary}
In our proof of \cref{cor:c=infty}, we will need the following basic fact.
\begin{fact}\label{fact:trace}
    In the setting of \cref{cor:c=infty} we have $\E[\mu_1+\dots+\mu_n]=mn$.
\end{fact}
\begin{proof}
    Let $G\in \mb C^{n\times m}$ be a random matrix with independent standard complex Gaussian entries, in such a way that $\mu_1,\dots,\mu_n$ are the eigenvalues of $GG^\dagger$. Then, we have
    \[\mu_1+\dots+\mu_n=\sum_{i=1}^n\sum_{j=1}^m |G_{ij}|^2,\] and the desired result follows (recalling that the entries $G_{ij}$ are complex standard Gaussian).
\end{proof}
\begin{proof}[Proof of \cref{cor:c=infty}]
Throughout this proof, all implicit constants in asymptotic notation are allowed to depend on $\ell$ (i.e., we think of $\ell$ as a constant).

 Writing $\Zt_i = X_i - \E X_i$, by \cref{thm:general-CLT-infinite}\eqref{item:CLT-infinite} we have the convergence in distribution $(\Zt_i)_{i \leq 2\ell} \to (Z_i)_{i \leq 2\ell}.$ Then, writing $T=\mu_1+\dots+\mu_n$, note that \begin{equation*}
       T = mn + 2 \sqrt{nm}\Zt_1,
    \end{equation*}
    using \cref{fact:trace}. 
    We may then write
    \begin{equation}
            Y_{2i} = T^{-2i} \sum_{k=1}^n \left(\mu_k - \frac{T}{n} \right)^{2i} = \left(\frac{4}{mn}\right)^i\left(1 + \frac{2 \Zt_1}{\sqrt{mn}} \right)^{-2i} \sum_{k=1}^n \left(\frac{\mu_k - m}{2\sqrt{mn}} - \frac{\Zt_1}{n}\right)^{2i}.\label{eq:cor-inf-1}
    \end{equation}
    Since $\tilde{Z_1}$ is bounded in probability by \cref{thm:general-CLT-infinite}\eqref{item:CLT-infinite}, we have that
    \begin{equation}
        \left(1 + \frac{2 \Zt_1}{\sqrt{mn}} \right)^{-2i} = 1 + O\left(\frac{1}{\sqrt{mn}}\right) = 1 + o\left(\frac 1n\right),\label{eq:cor-inf-2}
    \end{equation}
    where asymptotics are in probability. 

    Moreover, writing $\eps = -\Zt_1/n$, we expand
\begin{align*}
\sum_{k=1}^n \left(\frac{\mu_k - m}{2\sqrt{mn}} + \eps \right)^{2i} &= \sum_{k=1}^n \sum_{j = 0}^{2i} \binom{2i}{j} \left(\frac{\mu_k - m}{2\sqrt{mn}} \right)^{2i - j}\eps^j \\
         &= \sum_{j = 0}^{2i}\binom{2i}{j} X_{2i - j} \eps^j = \sum_{j = 0}^{2i}\binom{2i}{j} \left( \E X_{2i - j} + \Zt_{2i - j}\right) \eps^j\\
&= \E X_{2i} + \Zt_{2i}
  + \sum_{j = 1}^{2i}\binom{2i}{j}\Zt_{2i-j}\eps^j
  + \sum_{j = 1}^{2i}\binom{2i}{j}(\E X_{2i-j})\eps^j.
\end{align*}
Now, recall from \cref{thm:general-CLT-infinite}\eqref{item:LLN-infinite} that for each $r\le 2\ell$ we have
\[
\E X_r = n\gamma(x^r) + o(n).
\]
In particular, for each $j\le \ell$ we have $\E X_{2j} = O(n)$, while $\E X_{2j-1} = o(n)$ since $\gamma(x^{2j-1})=0$. Also, by \cref{thm:general-CLT-infinite}\eqref{item:CLT-infinite}, each $\Zt_r$ is bounded in probability, and hence $\eps = O(1/n)$ in probability.

Therefore every term in the two sums above is $o(1)$ in probability. Indeed, if $j\ge 1$, then
\[
\Zt_{2i-j}\eps^j = O(1)\cdot O(1/n^j)=o(1)
\]
in probability. Also, if $j=1$ and $2i-j$ is odd, then
\[
(\E X_{2i-1})\eps = o(n)\cdot O(1/n)=o(1)
\]
in probability, while if $j\ge 2$, then
\[
(\E X_{2i-j})\eps^j = O(n)\cdot O(1/n^j)=o(1)
\]
in probability. Since there are only finitely many such terms (with $\ell$ fixed), it follows that
\begin{equation}
\sum_{k=1}^n \left(\frac{\mu_k - m}{2\sqrt{mn}} + \eps \right)^{2i} = \E X_{2i} + \Zt_{2i} + o(1),\label{eq:cor-inf-3}
\end{equation}
where the $o(1)$ term is in probability. The desired result then follows from Equations \cref{eq:cor-inf-1,eq:cor-inf-2,eq:cor-inf-3}.
\end{proof}

Finally, as an analogue of \cref{lem:technical}, we need the following technical estimates on $\Gamma$, which are proved in \cref{sec:integral-estimates}. 

\begin{lemma}\label{lem:technical-infinite}
     $\phantom.$
    \begin{enumerate}
        \item The limit $\alpha=\lim_{i \to \infty} \Gamma(x^i,x^i)$ exists, and satisfies $\alpha>0$.\label{item:diagonal-infinite}
        \item We have $\lim_{A \to \infty} \lim_{i \to \infty} \Gamma(x^i, x^{Ai}) = 0. $ \label{item:off-diagonal-infinite}
    \end{enumerate}    
\end{lemma}

We can now complete the proof of \cref{thm:Nielsen-RMT}, supplementing the approximately balanced case in the last subsection with similar considerations in the imbalanced case (using \cref{cor:c=infty,lem:technical-infinite}).

\begin{proof}[Proof of \cref{thm:Nielsen-RMT}]
As before, let $\mathcal E_{\mr{maj}}$ be the event that $\lambda_k^{(1)}+\dots+\lambda_n^{(1)}\le \lambda_k^{(2)}+\dots+\lambda_n^{(2)}$ for all $k \leq n$. We wish to show that $\mb{P}[\mc{E}_{\mr{maj}}] \to 0$ as $m,n \to \infty$.
    
First, note that we can restrict our attention to the case where $m\ge n$. Indeed, recall that the complex Wishart--Laguerre ensemble with parameters $n,m$ is the distribution of a random matrix of the form $A A^\dagger$, where $A\in \mb C^{n\times m}$ is a matrix with standard complex Gaussian entries. Switching the roles of $m$ and $n$ is equivalent to considering the matrix $A^\dagger A$, which has the same nonzero eigenvalues as $AA^\dagger$ (if $n\le m$, this switch merely introduces $m-n$ additional zero eigenvalues) and this is the only relevant information for the event $\mc{E}_{\mr{maj}}$. 
    
Second, note that if \cref{thm:Nielsen-RMT} were not true, then there would be some sequence of parameter-pairs $(n,m_n)$, with $m_n \geq n\to \infty$, such that $\limsup_{n\to \infty}\Pr[\mathcal E_{\rm{maj}}]>0$. By compactness of $[1,\infty]$, there would then be an infinite subsequence of values of $n$ along which $m_n/n\to c$ for some $c\in [1,\infty]$. So, in proving \cref{thm:Nielsen-RMT} we may assume that $m_n/n\to c$ for some $c\in [1,\infty]$. In \cref{subsec:c<inf} we have already handled the case where $c<\infty$, we may (and do) assume that $m/n\to \infty$.

For $s \in \{1,2\}$, let $Y^{(s)}_{i} = \sum_{k=1}^{n}(\lambda_k^{(s)} - 1/n)^{i}$. By convexity of the functions $x \mapsto (x-1/n)^{2i}$, it follows from \cref{cor:c=infty,thm:convex-majorisation} (as in the approximately balanced case) that
\[\limsup_{n\to \infty}\mb{P}[\mc{E}_{\mr{maj}}] \leq \mb{P}\big[Z_{2i}^{(1)} \leq Z_{2i}^{(2)} \text{ for all } i\big],\]
where $(Z_{i}^{(1)})_i$ and $(Z_{i}^{(2)})_i$ are independent vectors distributed as in \cref{cor:c=infty}. 

Proceeding as in the approximately balanced case, but now using \cref{lem:technical-infinite}, it follows that if we let $(\zeta_j^{(1)}, \zeta_j^{(2)})$ denote i.i.d.\ (real) Gaussians with mean zero and variance $\alpha$, then
\begin{equation*}
    \mb{P}[Z_i^{(1)} \leq Z_{i}^{(2)}\text{ for all }i]
    \leq \lim_{N \to \infty} \mb{P}\big[\zeta_j^{(1)} \leq \zeta_j^{(2)}\text{ for all }j\le N\big] = 0. \qedhere
\end{equation*}
    
\end{proof}

\section{The uniform measure} \label{sec:uniform-measure}
In this section we prove \cref{thm:uniform}, for the uniform measure on the simplex.
The starting point is that if $(X_{1},\dots,X_{n})$ is uniformly
distributed on the simplex, then the decreasing rearrangement $(\lambda_{1},\dots,\lambda_{n})$
has a \emph{closed form description} in terms of independent random
variables, as follows.
\begin{theorem}
\label{fact:simplex-exponential}Let $Z_{1},\dots,Z_{n}$ be i.i.d.
exponential random variables with mean 1. Then the normalised vector
\[
\frac{1}{Z_{1}+\dots+Z_{n}}(Z_{1},\dots,Z_{n})
\]
is uniformly distributed in the simplex $\Delta_{n-1}$.
\end{theorem}

\begin{theorem}
\label{fact:exponential-order-statistics}Let $Z_{1},\dots,Z_{n}$
be i.i.d. exponential random variables with mean 1, and let $Z_{(1)}<\dots<Z_{(n)}$
be their order statistics (i.e., $(Z_{(n-i+1)})_{i=1}^{n}$ is the
decreasing rearrangement of $(Z_{i})_{i=1}^{n}$). Then $(Z_{(i)})_{i=1}^{n}$
has the same distribution as 
\[
\left(\frac{Z_{1}}{n},\quad\frac{Z_{1}}{n}+\frac{Z_{2}}{n-1},\quad\frac{Z_{1}}{n}+\frac{Z_{2}}{n-1}+\frac{Z_{3}}{n-2},\quad\dots,\quad\frac{Z_{1}}{n}+\dots+\frac{Z_{n}}{1}\right).
\]
\end{theorem}

\cref{fact:exponential-order-statistics,fact:simplex-exponential} are both classical. In particular, \cref{fact:simplex-exponential} is a direct
consequence of \cite[Chapter~V, Theorems~2.1 and 2.2]{Dev86}, and \cref{fact:exponential-order-statistics} can be found in \cite[Chapter~V, Section~3.3]{Dev86}.

Now, before getting into the details, we briefly describe the plan to prove \cref{thm:uniform}. First, by a standard concentration
inequality (e.g.\ Bernstein's inequality, see for example \cite[Theorem~2.8.1 and Example~2.7.8]{RV73}),
sums of exponential random variables do not fluctuate very much, as follows.
\begin{theorem}
\label{fact:bernstein}Let $Z_{1},\dots,Z_{N}$
be i.i.d. exponential random variables with mean 1, and let $U=a_{1}Z_{1}+\dots+a_{N}Z_{N}$ for some coefficients $a_1,\dots,a_N\in \mb R$. 
There is an absolute constant $c>0$ such that
\[
\Pr\big[|U-\mb EU|\ge t\big]\le\exp\left(-c\min\left(\frac{t^{2}}{a_{1}^{2}+\dots+a_{N}^{2}},\;\frac{t}{\max_{i}|a_{i}|}\right)\right)
\]
for every $t\ge0$.
\end{theorem}

Given \cref{fact:bernstein} (applied with $a_{1}=\dots=a_{n}=1$), we can approximate
quantities of the form $\lambda_{k}+\dots+\lambda_{n}$ (as in the
statement of \cref{thm:uniform}) with sums of the form 
\[
\sum_{j=1}^{n-k+1}\sum_{i=1}^{j}\frac{Z_{i}}{n-(i-1)}
\]
If $n-k$ is small, then each of the denominators $n-i+1$ is very
close to $n$ (i.e., the denominators are all almost the same), meaning
that it suffices to understand the \emph{iterated partial sums} of
the form $\sum_{j=1}^{q}\sum_{i=1}^{j}Z_{i}$. There is a fairly substantial
literature on this topic; in particular, we take advantage of the
following powerful theorem of Dembo, Ding and Gao~\cite[Theorem~1.1]{dembo2013persistence}, on \emph{persistence probabilities} for iterated partial sums.
\begin{theorem}
\label{thm:iterated-persistence}Let $(W_{i})_{i\in\mb N}$ be a sequence
of i.i.d. random variables with zero mean and $0<\mb E W_{i}^{2}<\infty$.
For $q\in\mb N$, let $S_{q}=\sum_{j=1}^{q}\sum_{i=1}^{j}W_{i}$. Then
for any fixed $t\in\mb R$ we have
\[
\Pr\big[S_{q}<t\text{ for all }q\le N\big]= \Theta_t(N^{-1/4}).
\]
\end{theorem}

We now provide the full details of the proof of \cref{thm:uniform}.
\begin{proof}
[Proof of \cref{thm:uniform}]
Let $(Z_j)_{j=1}^n$ and $(Z_j')_{j=1}^n$ be i.i.d.\ sequences of standard exponential random variables, and let $(Z_{(j)})_{j=1}^n$ and $(Z_{(j)}')_{j=1}^n$ be their order statistics. Let $\mc E_{\mr{maj}}$ be the event that \[\frac{Z_{(1)}+\dots+Z_{(q)}}{Z_{(1)}+\dots+Z_{(n)}}\le \frac{Z_{(1)}'+\dots+Z_{(q)}'}{Z_{(1)}'+\dots+Z_{(n)}'}\]
for all $q\le n$. Recalling \cref{fact:simplex-exponential}, our goal is to prove that $\Pr[\mc E_{\mr{maj}}]\le O((n/\log n)^{-1/16})$.

By \cref{fact:bernstein}, for a sufficiently large constant $C$ we have $|Z_1+\dots+Z_n-n|\le C\sqrt{n\log n}$ with probability at least $1 - n^{-100}$. So,
\[\Pr[\mc E_{\mr{maj}}]\le \Pr\left[Z_{(1)}+\dots+Z_{(q)}\le \Bigg(1+3C\sqrt{\frac{\log n}{n}}\,\Bigg) (Z_{(1)}'+\dots+Z_{(q)}')\text{ for all }q\le n\right]+n^{-100}.\]
Recalling \cref{fact:exponential-order-statistics}, it suffices to prove that
\begin{align} \label{eq:sufficient-for-uniform}
    \P\left[ \sum_{j = 1}^q \sum_{i = 1}^j \frac{Z_i}{n-(i-1)} \leq \Bigg(1 + 3C\sqrt{\frac{\log n}n}\,\Bigg)\sum_{j = 1}^q \sum_{i = 1}^j \frac{Z_i'}{n-(i-1)} \text{ for all }q\le n \right] \leq O((n/\log n)^{-1/16}).
\end{align}

Let $Q=\lfloor c(n/\log n)^{1/4}\rfloor$ for some small constant
$c>0$. By \cref{fact:bernstein}, with probability at least $1-n^{-100}$ we have
\[
\sum_{j=1}^{Q}\sum_{i=1}^{j}\frac{Z_{i}'}{n-(i-1)}=\sum_{i=1}^{Q}\bigg(\frac{Q-(i-1)}{n-(i-1)}\cdot Z_{i}'\bigg)\le\frac{2Q^{2}}{n}\le\frac{2c^{2}}{\sqrt{n\log n}},
\]
and similarly, with probability at least $1-n^{-100}$ we have
\[
\left|\sum_{j=1}^{k}\sum_{i=1}^{j}\frac{Z_{i}-Z_{i}'}{n-(i-1)}-\sum_{j=1}^{k}\sum_{i=1}^{j}\frac{Z_{i}-Z_{i}'}{n}\right|=\left|\sum_{i=1}^{k}\frac{i-1}{n}\cdot\frac{k-(i-1)}{n-(i-1)}\cdot(Z_{i}-Z_{i}')\right|\le\frac{1}{2n}
\]
for all $k\le Q$. Assuming that $c$ is sufficiently small (in terms of $C$),
it follows that the probability in \cref{eq:sufficient-for-uniform} is at most
\begin{align*}
 & \Pr\left[\sum_{j=1}^{q}\sum_{i=1}^{j}\frac{Z_{i}}{n-(i-1)}\le\sum_{j=1}^{q}\sum_{i=1}^{j}\frac{Z_{i}'}{n-(i-1)}+3C\sqrt{\frac{\log n}{n}}\cdot\frac{2c^{2}}{\sqrt{n\log n}}\text{ for all }q\le Q\right]+n^{-100}\\
 & \qquad\le\Pr\left[\sum_{j=1}^{q}\sum_{i=1}^{j}\frac{Z_{i}}{n-(i-1)}\le\sum_{j=1}^{q}\sum_{i=1}^{j}\frac{Z_{i}'}{n-(i-1)}+\frac{1}{2n}\text{ for all }q\le Q\right]+n^{-100}\\
 & \qquad\le\Pr\left[\sum_{j=1}^{q}\sum_{i=1}^{j}\frac{Z_{i}-Z_i'}{n}\le\frac{1}{n}\text{ for all }q\le Q\right]+2n^{-100}\\
 & \qquad=\Pr\left[\sum_{j=1}^{q}\sum_{i=1}^{j}(Z_{i}-Z_{i}')\le1\text{ for all }q\le Q\right]+2n^{-100}.
\end{align*}
This probability is at most $O(Q^{-1/4})=O((n/\log n)^{-1/16}),$ by \cref{thm:iterated-persistence}.
\end{proof}

\section{Approximate LOCC-convertibility}

In this section, we prove \cref{thm:approximate-conversion} in the following, more precise, form. 

\begin{theorem}
\label{thm:approximate-conversion-quantitative}
Let $M^{(1)},M^{(2)}\in \mb C^{n\times n}$ be independent samples from the trace-normalised complex Wishart--Laguerre ensemble with parameters $n,m$. Let $\lambda_1^{(1)}\ge \dots\ge \lambda_n^{(1)}$ and $\lambda_1^{(2)}\ge \dots\ge \lambda_n^{(2)}$ denote the respective eigenvalues, and let
    \[\Pi=\min_{1\le k\le n} \frac{\lambda_k^{(1)}+\lambda_{k+1}^{(1)}+\dots+\lambda_n^{(1)}}{\lambda_k^{(2)}+\lambda_{k+1}^{(2)}+\dots+\lambda_n^{(2)}}.\]
    For every  $\eps, \delta > 0$, there exists a constant $C(\eps,\delta) > 0$ such that if $m \geq n + C(\eps,\delta)\sqrt{n\log{n}}$ and $n \geq C(\eps,\delta)$, then $\mb{P}[\Pi < 1 - \eps] \leq \delta$.
\end{theorem}

In our proof of \cref{thm:approximate-conversion-quantitative} we need a few general tools. First, we need concentration of the trace of the complex Wishart--Laguerre ensemble.
\begin{lemma}\label{lem:trace-concentration}
    There is an absolute constant $c$ such that the following holds. Let $M$ be a sample from the complex Wishart--Laguerre distribution with parameters $n,m$. Then \[\Pr\big[|\!\on{tr}(M)-mn|>t\big]\le \exp\!\Big(-c\min\big(t^2/(mn),t\big)\Big).\]
\end{lemma}
\begin{proof}
     As in the proof of \cref{fact:trace}, we can write \[\on{tr}(M)=\sum_{i=1}^n\sum_{j=1}^m |G_{ij}|^2,\]
     where the $G_{ij}$ are independent standard complex Gaussians. The terms $|G_{ij}|^2$ are sub-exponential (see \cite[Lemma 2.7.6]{vershynin2018high}), so the desired result follows from Bernstein's inequality (see \cite[Theorem~2.8.1]{vershynin2018high}).
\end{proof}
We also need the Gaussian concentration inequality for Lipschitz functions, due independently to Borell and Sudakov--Tsirelson (see \cite[Definition~2.5.6 and Theorem~5.2.2]{vershynin2018high}).
\begin{theorem}\label{lem:gaussian-concentration}
There is an absolute constant $c>0$ such that the following holds. Let $F:\mb R^n\to \mb R$ be a Lipschitz function (with respect to the Euclidean metric), with Lipschitz constant at most $r$. Let $X_1,\dots,X_n\in \mb R$ be i.i.d.\ standard Gaussian random variables, and let $Y=F(X_1,\dots,X_n)$. Then for any $t\ge 0$ we have
\[\Pr\big[|Y-\mb E Y|\ge t\big]\le 2\exp(-ct^2/r^2).\]
\end{theorem}
For a matrix $G \in \mb{C}^{n\times m}$, let $\sigma_k(G)$ denote the $k$-th largest singular value of $G$ and let $\|G\|_\mr{F} = \sqrt{\sum_{ij}|G_{ij}|^2}$ denote the Frobenius norm of $G$. The next ingredient we will need is a version of Weyl's inequality. The following statement is a consequence of e.g.\ \cite[Exercise~1.3.22(iv)]{Tao12}, using the inequality $\|H\|_{\mr{F}} = (\mr{tr}(H^\dagger H))^{1/2} = (\sum_{i}\sigma_i(H)^2)^{1/2} \geq \sigma_1(H)$.
\begin{theorem}\label{lem:weyl}
    For any matrices $G,H\in \mb C^{n\times m}$, we have
    \[|\sigma_k(G+H) - \sigma_k(G)| \leq \|H\|_\mr{F}.\]
\end{theorem}

Finally, we need a lower bound on the expected least singular value of complex Gaussian rectangular matrices. The following estimate is a direct consequence of \cite[Theorem~3.4]{luh2020eigenvector}, which is an adaptation of the main result of \cite{RV09} to the complex case.
\begin{lemma}\label{lem:singular-value-lower-bound}
There is an absolute constant $c > 0$ such that the following holds. Let $G$ be an $n\times m$ matrix with independent complex standard Gaussian entries. Then, for all $1 \leq k \leq n$,
   \begin{align}
   \label{eq:sing-value-lb}
   \mb{E}[\sigma_k(G)] \geq \mb{E}[\sigma_n(G)] \geq c\big(\sqrt{m} - \sqrt{n-1}\big).
   \end{align}
\end{lemma}
Now we are ready to prove \cref{thm:approximate-conversion-quantitative}.
\begin{proof}[Proof of \cref{thm:approximate-conversion-quantitative}]
Throughout this proof, we treat $\varepsilon,\delta$ as constants. In particular, implicit constants in asymptotic notation are allowed to depend on $\varepsilon,\delta$. Our goal is to prove that $\mb{P}[\Pi < 1 - \eps] \leq \delta$, as long as $n$ is sufficiently large and $m\ge n+C\sqrt{n\log n}$ for large enough $C$.

   Let $G^{(1)}, G^{(2)}$ be independent $n\times m$ matrices with independent complex standard Gaussian entries, and for $s\in\{1,2\}$ let $M^{(s)} = G^{(s)}(G^{(s)})^{\dagger} \in \mb{C}^{n\times n}$ denote the corresponding sample from the complex Wishart--Laguerre ensemble with parameters $n,m$. Let $\mu^{(s)}_1 \geq \dots \geq \mu^{(s)}_n \geq 0$ denote the eigenvalues of $M^{(s)}$; we may take $\lambda^{(s)}_1 \geq \dots \geq \lambda^{(s)}_n \geq 0$ to be the eigenvalues of $M^{(s)}/\on{tr}(M^{(s)})$.

   By \cref{lem:trace-concentration},
   with probability at least $1-\delta/2$ we have $\on{tr}(M^{(s)}) = mn + O(\sqrt{mn})$ for both $s \in \{1,2\}$. Therefore, it suffices to show that $\mb{P}[\Pi_{\mu} < 1 - \eps] \leq \delta/2$, where
   \[\Pi_{\mu} = \min_{1\le k\le n} \frac{\mu_k^{(1)}+\mu_{k+1}^{(1)}+\dots+\mu_n^{(1)}}{\mu_k^{(2)}+\mu_{k+1}^{(2)}+\dots+\mu_n^{(2)}}.\]
    We will, in fact, show the stronger statement that
    \begin{align*}
        \mb{P}\left[\mu_k^{(1)} \in [(1-\eps)\mu_k^{(2)}, (1+\eps)\mu_k^{(2)}] \quad \text{for all }1 \leq k \leq n\right] \geq 1-\delta/2. 
    \end{align*}
    By the union bound, it suffices to show that 
    \begin{align}
    \label{eqn:eigenvalue-concentration}
        \mb{P}\left[\mu_k^{(1)} \in [(1-\eps)\mu_k^{(2)}, (1+\eps)\mu_k^{(2)}]\right] \geq 1-\delta/(2n) \quad \text{for all }1 \leq k \leq n. 
    \end{align}
   For $s\in\{1,2\}$, let $\sigma^{(s)}_k=\sqrt{\mu^{(s)}_k}$ be the $k$-th largest singular value of $G^{(s)}$. We can interpret $\sigma_k^{(s)} : \mb{C}^{n\times m} \to \mb{R}$ as being a function of $2mn$ i.i.d.\ standard (real) Gaussian random variables (the real part and imaginary part of each entry of $G^{(s)}$), and by Weyl's inequality (\cref{lem:weyl}),
   this function is $1$-Lipschitz with respect to the Euclidean metric on $\mb{C}^{n\times m}\cong \mb R^{2nm}$. It follows from the Gaussian concentration inequality (\cref{lem:gaussian-concentration})
   that for any $1 \leq k \leq n$ (and $s\in \{1,2\}$),
   \begin{align}
   \label{eqn:gauss-concentration}
   \mb{P}\left[\big|\sigma_k^{(s)} - \mb{E}[\sigma_k^{(s)}]\big| \geq t\right] \leq 2\exp(-ct^2).
   \end{align}
   Note that our assumption $m\ge n+C\sqrt{n\log n}$, together with \cref{lem:singular-value-lower-bound}, yields
   \begin{align}
   \label{eqn:sing-value-lb}
   \mb{E}[\sigma_k^{(s)}] \ge C'\sqrt{\log n},
   \end{align}
   where $C'$ can be made arbitrary large by taking large enough $C$.
   
   We now prove \cref{eqn:eigenvalue-concentration}. Applying \cref{eqn:gauss-concentration} with $t = 10\sqrt{\log{(n/\delta)}}=O(\sqrt{\log n})$, and using \cref{eqn:sing-value-lb}, we see that with probability at least $1-\delta/(2n)$, for $s\in \{1,2\}$ we have $\sigma_k^{(s)}=\mb E \sigma_k^{(s)}+O(\sqrt{\log n})=\mb E \sigma_k^{(s)}(1+O(1/C'))$. Since $\mb E \sigma_k^{(1)}=\mb E \sigma_k^{(2)}$, this implies that
   \[\mu_k^{(1)}=(\sigma_k^{(1)})^2=\Big(\sigma_k^{(2)}\big(1+O(1/C')\big)\Big)^2=\mu_k^{(2)}(1+O(1/C')).\]
   The desired result follows, assuming $C'$ is sufficiently large (which we may ensure by taking $C$ sufficiently large in terms of $\varepsilon$).
\end{proof}
   
\section{Integral estimates}\label{sec:integral-estimates}
In this section we prove the estimates in \cref{lem:technical} and \cref{lem:technical-infinite}, via a sequence of estimates of various integrals. In this section we write $f\lesssim g$ to mean $f=O(g)$.

First, for the convenience of the reader we recall the relevant parts of \cref{def:CLT-prep,def:CLT-prep-2}.
\begin{definition}
For functions $f,g:\mb R\to \mb R$, let
\begin{align*}
\gamma(f) &= \frac{1}{\pi}\int_{-1}^1 f(x) \sqrt{1 - x^2}\,dx,\\
\Gamma(f,g)&=\frac{1}{\pi^2} \int_{-1}^{1}\int_{-1}^{1} \left(\frac{f(x) - f(y)}{x-y} \right) \left(\frac{g(x) - g(y)}{x-y} \right) \frac{1 - xy}{\sqrt{1 - x^2} \sqrt{1 - y^2}}\,dx\,dy.
\end{align*}
Then, for $a_{\pm} = (1 \pm \sqrt{c})^2$ and $\Phi_c(x) = 2\sqrt{c} x + c + 1$, let
    $$\gamma_c(f) = \frac{1}{2\pi} \int_{a_-}^{a_+} f(x)\frac{\sqrt{(a_+ - x)(x - a_-)}}{x} \,dx,\qquad \Gamma_c(f,g) = \Gamma(f\circ \Phi_c,g\circ \Phi_c).$$ 
\end{definition}

Now, to begin with, we show that the kernel appearing in the definition of $\Gamma$ in fact yields a density of a probability measure on $[-1,1]^2$.

\begin{fact}\label{lem:this-is-density}
We have
\[\frac{1}{\pi^2}\int_{-1}^1 \int_{-1}^1 \frac{1 - xy}{\sqrt{(1 - x^2)(1 - y^2)}}\,dx\,dy=1.\]
\end{fact}
\begin{proof}
With the substitution $(s,t)=(1-x,1-y)$ and the symmetry between $s$ and $t$, we compute
\begin{align*}
    &\frac{1}{\pi^2}\int_{-1}^1 \int_{-1}^1 \frac{1 - xy}{\sqrt{(1 - x^2)(1 - y^2)}}\,dx\,dy \\
    &\qquad= \frac{1}{\pi^2}\int_0^2 \int_0^2 \frac{s + t - ts}{\sqrt{st(2 - s)(2 - t)}}\,ds\,dt \\
    &\qquad= \frac{2}{\pi^2}\int_0^2 \int_0^2 \frac{\sqrt{s}}{\sqrt{t(2 - s)( 2- t)}}\,ds\,dt - \frac{1}{\pi^2}\int_0^2 \int_0^2  \frac{\sqrt{st}}{\sqrt{(2 - s)(2 - t)}}\,ds\,dt =1,
\end{align*} 
where we used the identities $$\int_0^2 \sqrt{\frac{s}{2 - s}}\,ds = 2\int_0^1 \frac{x^{1/2}}{(1 - x)^{1/2}}\,dx=\pi,\quad \int_0^2 \frac{1}{\sqrt{s(2 - s)}}\,ds = \int_0^1 \frac{x^{-1/2}}{(1 - x)^{1/2}} = \pi$$ (these are Beta integrals; see for example \cite[Eqns.~(6.1.6) and~(6.1.7)]{Mol15}).
\end{proof}

Now we proceed to study the asymptotics of the quantities appearing in \cref{lem:technical,lem:technical-infinite}.

\begin{lemma}\label{lem:gamma-limit}
    Fix $c \in [1,\infty)$, and let $h_k(x) = (x/a_+)^k.$  There is a constant $C_c \in (0,\infty)$ so that $$\lim_{k \to \infty} k^{3/2} \gamma_c(h_k) = C_c. $$
\end{lemma}
\begin{proof}
    It will be more convenient to work with $h_{k+1}$.  By changing variables $x = c + 1 + (2 \sqrt{c})t$, we see \begin{align*}
        \gamma_c(h_{k+1}) &= \frac{1}{2\pi} \int_{a_-}^{a_+} \bigg(\frac x{a_+}\bigg)^{k+1}\frac{\sqrt{(a_+ - x)(x - a_-)}}{x} \,dx\\
        &=\frac{2c}{\pi a_+} \int_{-1}^1 \left(1 - (1 - t)\frac{2\sqrt{c}}{c + 1 + 2\sqrt{c}} \right)^k \sqrt{(1 - t)(t + 1)}\,dt.
    \end{align*}
    Writing $\beta = \frac{2\sqrt{c}}{c + 1 + 2\sqrt{c}}$ and changing variables by $1  -t = s/k$, we have \begin{align*}
        \gamma_c(h_{k+1}) =  \frac{2 c k^{-3/2}}{\pi a_+}\int_0^{2k} \left(1 - \frac{\beta s}{k} \right)^k \sqrt{s(2 - s/k) }\,ds.
    \end{align*}
    We note that $(1 - \beta s /k)^{k} \leq e^{-\beta s}$ and so by the dominated convergence theorem, we have \begin{align*} \lim_{k \to \infty}
        \int_0^{2k} \left(1 - \frac{\beta s}{k} \right)^k \sqrt{s(2 - s/k) }\,ds = \sqrt{2} \int_0^\infty e^{-\beta s} \sqrt{s}\,ds = \frac{\sqrt{\pi}}{\sqrt{2} \beta^{3/2}}
    \end{align*}
    (the last integral is a Gamma integral; see for example \cite[Eqn.~(4.1.7)]{Mol15}). The desired result follows.
\end{proof}

For the next estimate, we need the following simple inequality.

\begin{fact}\label{fact:MVT-for-exp}
    For $1 \geq b > a \geq 0$ and $k >  0$ we have $$(1 - a)^k - (1 - b)^k \leq e(e^{-ak} - e^{-bk})$$
\end{fact}
\begin{proof}
    Bound \begin{equation*}(1 - a)^k - (1 - b)^k = \int_a^b k (1 - \theta)^{k-1} d\theta \leq \int_a^{b} k e^{-\theta(k-1)}\,d\theta \leq e( e^{-ak} - e^{-bk}).\qedhere \end{equation*}
\end{proof}

\begin{lemma}\label{lem:cov-limit-finite-c}
    For each fixed $c,A \in [1,\infty)$, and  $h_k(x) = (x/a_+)^k$ for $k \geq 1$, we have $$\lim_{k \to \infty} \Gamma_c(h_k,h_{Ak}) = \frac{1}{2\pi^2} \int_0^\infty \int_0^\infty \left(\frac{e^{-s} - e^{-t}}{s - t} \right)\left(\frac{e^{-As} - e^{-At}}{s - t} \right)\frac{s + t}{\sqrt{st}} \,ds\,dt.$$
\end{lemma}
\begin{proof}
    Writing $\beta = 2\sqrt{c}/a_+$ and changing variables $(x,y) = \big(1 - s/(\beta k),1 - t/(\beta k)\big)$, recalling that $\Phi_c(x)=2\sqrt c x+c+1$, we see that \begin{align*}
        \Gamma_c(h_k,h_{Ak})&=\frac{1}{\pi^2} \int_{-1}^{1}\int_{-1}^{1} \left(\frac{(\Phi_c(x)/a_+)^k - (\Phi_c(y)/a_+)^k}{x-y} \right) \left(\frac{(\Phi_c(x)/a_+)^{Ak} - (\Phi_c(y)/a_+)^{Ak}}{x-y} \right)\\
        &\qquad\qquad\qquad\qquad\qquad\cdot\frac{1 - xy}{\sqrt{1 - x^2} \sqrt{1 - y^2}}\,dx\,dy \\
        &= \frac{1}{\pi^2} \int_0^{2\beta k}\int_0^{2\beta k} \left(\frac{(1 - s/k)^k - (1 - t/k)^k}{s -t} \right)\left(\frac{(1 - s/k)^{Ak} - (1 - t/k)^{Ak}}{s -t} \right)\\
        &\qquad\qquad\qquad\qquad\qquad\cdot\frac{s + t - st/(\beta k)}{\sqrt{st(2 - s/(\beta k))(2 - t/(\beta k))}} \,ds\,dt.
    \end{align*}    
    We may bound the integrand using \cref{fact:MVT-for-exp} by a constant multiple of $$ \left(\frac{e^{-s} - e^{-t}}{s - t} \right)\left(\frac{e^{-As} - e^{-At}}{s - t} \right)\frac{s + t}{\sqrt{st}}$$
    which is integrable (see \cref{lem:A-limit}, to follow), so the dominated convergence theorem completes the proof.
\end{proof}

\begin{lemma}\label{lem:cov-limit-infinite-c}
    For  $ A \ge 1$, we have $$\lim_{k \to \infty} \Gamma(x^k,x^{Ak}) = \frac{1}{\pi^2} \int_0^\infty \int_0^\infty \left(\frac{e^{-s} - e^{-t}}{s - t} \right)\left(\frac{e^{-As} - e^{-At}}{s - t} \right)\frac{s + t}{\sqrt{st}} \,ds\,dt.$$
\end{lemma}
\begin{proof}
    We first claim that \begin{equation} \label{eq:cross-sign}
    \lim_{k \to \infty} \int_0^1 \int_{-1}^0 \left(\frac{x^k - y^k}{x - y} \right)\left(\frac{x^{Ak} - y^{Ak}}{x - y} \right) \frac{1 - xy}{\sqrt{(1 - x^2)(1 - y^2)}}\,dx\,dy = 0.
    \end{equation}
    To see this, note that on the set $[0,1] \times [-1,0]$ we may bound \begin{align*}
        \left(\frac{x^k - y^k}{x - y} \right)\left(\frac{x^{Ak} - y^{Ak}}{x - y} \right) \leq \begin{cases}
            16 & \text{ if either } x \geq 1/2 \text{ or } y \leq -1/2 \\
            A k^2 2^{-2(k-1)} & \text{ if } |x| \leq 1/2 \text{ and } |y| \leq 1/2 
        \end{cases}
    \end{align*} 
    by the mean-value theorem.  Applying the dominated convergence theorem shows \eqref{eq:cross-sign}.  By symmetry of the integrand under $(x,y) \mapsto -(x,y)$ this shows that $$\lim_{k \to \infty} \Gamma(x^k,x^{Ak}) =  \lim_{k \to \infty} \frac{2}{\pi^2}\int_0^1 \int_0^1 \left(\frac{x^k - y^k}{x - y} \right)\left(\frac{x^{Ak} - y^{Ak}}{x - y} \right) \frac{1 - xy}{\sqrt{(1 - x^2)(1 - y^2)}}\,dx\,dy.$$
    Changing variables $(x,y) = (1 - s/k,1-t/k)$ and arguing as in \cref{lem:cov-limit-finite-c} completes the proof. 
\end{proof}

\begin{lemma}\label{lem:A-limit}
    For $A \geq 1$ we have 
    $$ \int_0^\infty\int_0^\infty \left(\frac{e^{-s} - e^{-t}}{s - t}\right)\left(\frac{e^{-As} - e^{-At}}{s - t}\right)\frac{s+t}{\sqrt{st}}\,ds\,dt = O(A^{-1/2}(1+\log A)).$$
\end{lemma}
\begin{proof}
	Since the integrand is symmetric in $s$ and $t$ it is sufficient to integrate over the region $s < t$, in which case we may bound the integrand by a constant times \begin{align*}
	 \iint_{0 \leq s < t} \left(\frac{e^{-s} - e^{-t}}{s - t}\right)\left(\frac{e^{-As} - e^{-At}}{s - t}\right)\sqrt{\frac{t}{s}}\,ds\,dt.
	\end{align*}
	We will further break up the region of integration depending on the value of $t-s$. For $t-s \leq 1$, bound the integrand by \begin{align*}
	e^{-s}\left|\frac{e^{-As} - e^{-At}}{s - t} \right|(1 + s^{-1/2}).
\end{align*}
(Here we used that $(e^{-s}-e^{-t})/(s-t)\le e^{-s}$, which follows from the classical inequality $1-e^{x}\le x$ with $x=s-t$). Then, write $t = s + h$, yielding $$\iint_{0\leq s \leq t\le s+1}  \left(\frac{e^{-s} - e^{-t}}{s - t}\right)\left(\frac{e^{-As} - e^{-At}}{s - t}\right)\sqrt{\frac{t}{s}}\,ds\,dt \le \int_0^\infty e^{-s}(1 + s^{-1/2}) e^{-As}\int_0^1 \frac{1 - e^{-Ah}}{h}\,dh \,ds. $$
Now, observe that
\begin{align*}
		\int_0^1 \frac{1 - e^{-Ah}}{h}\,dh &= \int_0^{1/A} \frac{1 - e^{-Ah}}{h}\,dh  + \int_{1/A}^1 \frac{1 - e^{-Ah}}{h}\,dh \\
		&\leq \int_0^{1/A} A \,dh + \int_{1/A}^1 h^{-1}\,dh  \\
		&= 1 + \log A,
	\end{align*}
and
\begin{align}\int_0^\infty e^{-s}(1 + s^{-1/2}) e^{-As}&\le \frac1{\sqrt{A}}+\sqrt{A+1}\int_{1/\sqrt A}^\infty ((A+1)s)^{-1/2} e^{-(A+1)s}\notag\\
&=\frac1{\sqrt{A}}+\frac{1}{\sqrt{A+1}}\int_{1}^\infty q^{-1/2} e^{-q}\lesssim  \frac{1}{\sqrt{A}}.\label{eq:1/sqrtA}
\end{align}
We deduce that \begin{align*}
	\iint_{0\leq s \leq t\le s+1}  \left(\frac{e^{-s} - e^{-t}}{s - t}\right)\left(\frac{e^{-As} - e^{-At}}{s - t}\right)\sqrt{\frac{t}{s}}\,ds\,dt 
 \lesssim  \frac{1+\log A}{\sqrt{A}}.
\end{align*}
This takes care of the region where $t-s\le 1$; we now turn to the complementary region $t-s\ge 1$.
With the substitution $q=(t-s)/s$ we first compute
\[\int_{s + 1}^\infty (t-s)^{-2} \sqrt{1 + \frac{(t-s)}{s}}dt=s\int_{1/s}^\infty q^{-2}\sqrt{1+q}\,dq\lesssim s^{-1/2}.\]
Combining this with \cref{eq:1/sqrtA}, we have
\begin{align*}
		\iint_{t \geq s + 1}  \left(\frac{e^{-s} - e^{-t}}{s - t}\right)\left(\frac{e^{-As} - e^{-At}}{s - t}\right)\sqrt{\frac{t}{s}}\,ds\,dt &\leq \int_0^\infty e^{-s(1 + A)} \int_{s + 1}^\infty (t-s)^{-2} \sqrt{1 + \frac{(t-s)}{s}} \,dt\,ds \\
		&\lesssim \int_0^\infty e^{-s(1 + A)} s^{-1/2} \,ds\lesssim \frac1{\sqrt A}. \qedhere
\end{align*}

\end{proof}

\begin{proof}[Proof of \cref{lem:technical}]
Note that (\cref{item:gamma}) follows from \cref{lem:gamma-limit}. The other two items follow by combining \cref{lem:cov-limit-finite-c} with \cref{lem:A-limit}.  
\end{proof}

\begin{proof}[Proof of \cref{lem:technical-infinite}]
Both items follow by combining \cref{lem:cov-limit-infinite-c} with the estimate from \cref{lem:A-limit}. 
\end{proof}

\newpage

\medskip\noindent\textbf{Data availability statement. }No datasets were used in this research.

\medskip\noindent\textbf{Conflicts of interest. }The authors have no conflicts of interest to declare that are relevant to the content of this
article

\bibliographystyle{amsplain_initials_nobysame_nomr}
\bibliography{main}

\appendix
\section{A general central limit theorem for linear statistics of the Wishart--Laguerre ensemble}\label{sec:CLT}
In this appendix we explain how to deduce the statements of \cref{thm:general-CLT,thm:general-CLT-infinite} from results in \cite{lytova2009central,BY88,Bao15}. As in the statements of these theorems, let $\mu_1,\dots,\mu_n$ be the eigenvalues of a random matrix sampled from the complex Wishart--Laguerre ensemble with parameters $n,m$, and assume $m/n\to c\in [1,\infty]$.


We recall the relevant statements from \cite{lytova2009central,BY88,Bao15}. Let $a_\pm = (1\pm \sqrt c)^2$, and let $\mathcal F_0$ be the set of all differentiable functions $f:\mb R\to \mb R$ such that both $f$ and $f'$ have compact support. The cited results are stated under slightly different regularity assumptions, but $\mathcal F_0$ is a convenient common subclass which certainly satisfies each of them. Below we will use truncation and concentration to pass from $\mathcal F_0$ to the larger class $\mc F$ appearing in \cref{thm:general-CLT,thm:general-CLT-infinite}; in particular, the polynomial test functions used in \cref{cor:c<infty,cor:c=infty} belong to $\mc F$. Fix $g\in \mc F_0$.

First, if $c<\infty$ then let $W=g(\mu_1/n)+\dots+g(\mu_n/n)$. In this case, \cite[Theorem~4.1]{lytova2009central} (which is really a restatement of results in \cite{MP67}) says that
\begin{equation}\frac Wn\overset p\to \frac{1}{2\pi}\int_{a_-}^{a_+} g(\lambda)\frac{\sqrt{(\lambda-a_-)(a_+-\lambda)}}\lambda d\lambda = \gamma_c(g),\label{eq:c=1-mean}\end{equation}
and \cite[Theorem~4.2]{lytova2009central} says that $W-\mb EW\overset d\to \mc N(0,\sigma^2)$, where
\begin{align}\sigma^2 &=\frac 12\cdot \frac{1}{2\pi^2}\int_{a_-}^{a_+}\int_{a_-}^{a_+}\left(\frac{g(\lambda_1) - g(\lambda_2)}{\lambda_1 - \lambda_2}\right)^2 \frac{4c - (\lambda_1 - c - 1)(\lambda_2 - c - 1)}{\sqrt{4c - (\lambda_1 - c - 1)^2}\sqrt{4c - (\lambda_2 - c - 1)^2}}\,d\lambda_1d\lambda_2 \nonumber \\
&= \Gamma_c(g,g).
\label{eq:c=1-variance}\end{align}

Second, if $c=\infty$, then with $X=g((\mu_1-m)/2\sqrt{mn})+\dots+g((\mu_n - m)/2\sqrt{mn})$, the convergence to the semicircle law proved in \cite{BY88} implies that 
\begin{equation}\frac Xn\overset p\to \frac{2}{\pi}\int_{-1}^{1}g(\lambda)\sqrt{1-\lambda^2}d\lambda = \gamma(g).\label{eq:c=infty-mean}\end{equation}
With the slightly different parameterisation $Y=g((\mu_1-m)/\sqrt{mn})+\dots+g((\mu_n - m)/\sqrt{mn})$, \cite[Theorem~1.2]{Bao15} says that $Y-\mb EY\overset d\to \mc N(0,\sigma^2)$, where
\begin{equation}\sigma^2=\frac 12\cdot \frac{1}{2\pi^2}\int_{-2}^{2}\int_{-2}^{2}\left(\frac{g(\lambda_1) - g(\lambda_2)}{\lambda_1 - \lambda_2}\right)^2 \frac{4 - \lambda_1\lambda_2}{\sqrt{4 - \lambda_1^2}\sqrt{4 - \lambda_2^2}}\,d\lambda_1d\lambda_2 = \Gamma(g,g).\label{eq:c=infty-variance}\end{equation}
\begin{remark*}
    The above cited papers all primarily concern the \emph{real} Wishart--Laguerre ensemble; we need the complex case. The same proofs apply to both cases (the only difference is that the variance in the real case is a factor of 2 larger than the variance in the complex case). See \cite[Remark~4.1]{lytova2009central}.
\end{remark*}

These results imply the univariate cases of \cref{thm:general-CLT,thm:general-CLT-infinite}, assuming that $f \in \mc{F}_0$ and provided we weaken $L^1$-convergence to convergence in probability. We now discuss how to upgrade this to the desired statements, using known results on concentration of singular values of Gaussian random matrices. Specifically, we will need the following.

\begin{lemma}
\label{lem:sv-concentration}
    Let $\mu_1 \geq \dots \geq\mu_n$ be the eigenvalues of a random matrix sampled from the complex Wishart--Laguerre ensemble with parameters $m\geq n$. There exists an absolute constant $C$ such that for any $t \geq 0$,
    \[\sqrt{m} - C(\sqrt{n}+t)\leq \sqrt{\mu_n} \leq \sqrt{\mu_1} \leq \sqrt{m} + C(\sqrt{n}+t)\]
    with probability at least $1 - 2\exp(-t^2)$. 
\end{lemma}

\begin{remark*}
    In the real case, this statement is proved in \cite[Theorem~4.6.1]{vershynin2018high}, but the same proof extends readily to the complex case as well. 
\end{remark*}
Therefore, in the case $c < \infty$, it follows from the upper bound in \cref{lem:sv-concentration} that for all sufficiently large $s$ (depending on $c$),
\begin{align*}
    \mb{P}\Big[\max_{k} \mu_k/n \geq s\Big] = \mb{P}[\sqrt{\mu_1} \geq \sqrt{sn}] \leq \mb{P}[\sqrt{\mu_1} \geq \sqrt{m} + C\sqrt{n} + C\sqrt{sn/2}] \leq 2\exp(-sn/4), 
\end{align*}
so that if $g \in \mc{F}$ (i.e.~$g$ is continuously differentiable and has sub-exponential growth) then for any $\varepsilon > 0$, we can (smoothly) truncate $g \in \mc{F}$ to some function $g_0 \in \mc{F}_0$ with compact support, in such a way that
\[\mb E\Big|\Big(g\big(\mu_1/n)\big)+\dots+g\big(\mu_n/n)\big)\Big)-\Big(g_0\big(\mu_1/n\big)+\dots+g_0\big(\mu_n/n\big)\Big)\Big|\le \varepsilon.\]
That is to say, the effect of the truncation can be made arbitrarily small, so if the asymptotic results in \cref{eq:c=1-mean,eq:c=1-variance} hold for every $g\in \mc F_0$ they must also hold for every $g\in \mc F$.

Further, by the upper bound in \cref{lem:sv-concentration} and using that $g \in \mc{F}$ has sub-exponential growth, it follows that the sequence of random variables $W/n$ is uniformly integrable. Since this sequence converges to $\gamma_c(g)$ in probability by \cref{eq:c=1-mean}, it follows by uniform integrability that it also converges in $L^1$. 

In the case $c = \infty$, we use both the upper and the lower bound in \cref{lem:sv-concentration} to see that for all sufficiently large $s$,
\begin{align*}
    \mb{P}\Big[\max_k (\mu_k - m)/\sqrt{mn} \geq s\Big] &\leq \mb{P}[\mu_1 \geq m + s\sqrt{mn}] + \mb{P}[\mu_n \leq m - s\sqrt{mn}]\\
    &\leq 2\min\{\exp(-s^2 n/2), \exp(-s\sqrt{mn}/2)\}\\
    &\leq 2\exp(-sn/4),
\end{align*}
from which it follows (as before) that we can upgrade the assumption $g \in \mc{F}_0$ to the assumption $g \in \mc{F}$ and also upgrade the convergence in \cref{eq:c=infty-mean} to convergence in $L^1$. 


Finally, we discuss how to deduce multivariate central limit theorems from the above univariate central limit theorems. Given functions $f_1,\dots,f_\ell\in \mc F$ and a vector $\vec t=(t_1,\dots,t_\ell)\in \mb R^\ell$, define
\[
f_{\vec t}=t_1f_1+\dots+t_\ell f_\ell.
\]
Applying the above univariate convergence to the function $f_{\vec t}$, we obtain in the case $c<\infty$ that
\[
t_1(W_1-\mb E W_1)+\dots+t_\ell(W_\ell-\mb E W_\ell)
\overset d\to \mc N\bigl(0,\Gamma_c(f_{\vec t},f_{\vec t})\bigr),
\]
and in the case $c=\infty$ that
\[
t_1(Y_1-\mb E Y_1)+\dots+t_\ell(Y_\ell-\mb E Y_\ell)
\overset d\to \mc N\bigl(0,\Gamma(f_{\vec t},f_{\vec t})\bigr).
\]
By bilinearity,
\[
\Gamma_c(f_{\vec t},f_{\vec t})=\sum_{a,b=1}^\ell t_a t_b\,\Gamma_c(f_a,f_b),
\qquad
\Gamma(f_{\vec t},f_{\vec t})=\sum_{a,b=1}^\ell t_a t_b\,\Gamma(f_a,f_b).
\]
Thus the covariance forms appearing here are exactly those associated with the matrices
\[
\Sigma_c=\bigl(\Gamma_c(f_a,f_b)\bigr)_{a,b\le \ell},
\qquad
\Sigma_\infty=\bigl(\Gamma(f_a,f_b)\bigr)_{a,b\le \ell}.
\]
The desired multivariate convergences in \cref{thm:general-CLT,thm:general-CLT-infinite} now follow from the Cram\'er--Wold device; see, for example, \cite[Theorem~29.4]{Bill95}.

\end{document}